\documentclass[letterpaper]{article} 
\usepackage{aaai2026}  

\usepackage{times}  
\usepackage{helvet}  
\usepackage{courier}  
\usepackage[hyphens]{url}  
\usepackage{graphicx} 
\usepackage{natbib}  
\usepackage{caption} 
\usepackage{algorithm}
\usepackage{algorithmic}

\usepackage{booktabs}       
\usepackage{amsfonts}       
\usepackage{nicefrac}       
\usepackage{microtype}      
\usepackage{xcolor}         
\usepackage{xspace}

\usepackage{subcaption}

\usepackage{algorithm}
\usepackage{algorithmic}
 
\usepackage{amssymb}
\usepackage{amsmath}
  
\usepackage{multirow}
\usepackage{amsthm}

\usepackage{bm}
\usepackage{threeparttable}

\usepackage{color, colortbl}
\definecolor{greyC}{RGB}{180,180,180}
\definecolor{greyL}{RGB}{235,235,235}
\definecolor{citeColor}{RGB}{0,20,115}

\def\mod{\mathbf{W}} 
\def\tmod{\mathbf{T}} 
\def\noi{\mathcal{N}} 
\def\ournoi{\tilde{\mathcal{N}}}
\def\noimatri{\mathbf{N}}

\def\pca{\mathcal{F}_{\mathsf{PCA}}}
\def\subsp{\mathbf{B}}
\def\subspall{\hat{\mathbf{B}}}
\def\nos{\mathsf{nos}}

\def\vsvd{\mathbf{V}}

\def\cla{\mathcal{C}}
\def\mecha{\mathcal{M}}
\def\aux{\mathsf{aux}}

\def\var{\mathsf{Var}}

\def\egv{\mathsf{egv}}
\def\svec{\mathsf{SVec}}
\def\nnet{\mathsf{NNet}}

\newtheorem{theorem}{Theorem}
\newtheorem{definition}{Definition}

\def\ie{\textit{i.e.}\xspace} 

\def\eg{\textit{e.g.}\xspace}

\def\etal{\textit{et~al.}\xspace}

\def\dphero{\mathsf{DP\text{-}Hero}}

\newcommand{\yuzh}[1]{{\color{magenta} #1}}

\usepackage{newfloat}
\usepackage{listings}
\DeclareCaptionStyle{ruled}{labelfont=normalfont,labelsep=colon,strut=off} 
\floatstyle{ruled}
\newfloat{listing}{tb}{lst}{}
\floatname{listing}{Listing}
\title{Revisiting Privacy-Utility Trade-off for DP Training with\\Pre-existing  Knowledge}
\author{
    Yu Zheng\textsuperscript{\rm 1}\thanks{Partial work was done when Yu was at Chinese University of Hong Kong and Hong Kong Polytechnic University.},
    Wenchao Zhang\textsuperscript{\rm 3},
    Yonggang Zhang\textsuperscript{\rm 2}\thanks{Corresponding author.},
    Wei Song\textsuperscript{\rm 3},
    Yuxiang Peng\textsuperscript{\rm 3},\\
    Kai Zhou\textsuperscript{\rm 4},
    Xiaojiang Du\textsuperscript{\rm 5},
    Bo Han\textsuperscript{\rm 6}
}
\affiliations{
    \textsuperscript{\rm 1}University of California, Irvine,
    \textsuperscript{\rm 2}Hong Kong University of Science and Technology,
    \textsuperscript{\rm 3}Northeastern University\\
    \textsuperscript{\rm 4}Hong Kong Polytechnic University,
    \textsuperscript{\rm 5}Stevens Institute of Technology,
    \textsuperscript{\rm 6}Hong Kong Baptist University

}

\usepackage{bibentry}

\begin{document}

\maketitle

\begin{abstract}
Differential privacy (DP) provides a provable framework for protecting individuals by customizing a random mechanism over a privacy-sensitive dataset.
Deep learning models have demonstrated privacy risks in model exposure as an established learning model unintentionally records membership-level privacy leakage.
Differentially private stochastic gradient descent (DP-SGD) has been proposed to protect training individuals by adding random Gaussian noise to gradient updates in backpropagation.
Researchers identified that DP-SGD causes utility loss, as the homogeneous noise injected can alter the gradient updates calculated at each iteration. Namely, all elements in the gradient are contaminated regardless of their importance in updating model parameters.
In this work, we argue that the utility can be optimized by involving the heterogeneity of the injected noise.
Consequently, we propose a generic \underline{d}ifferential \underline{p}rivacy framework with \underline{he}te\underline{ro}geneous noise ($\dphero$) by defining a heterogeneous random mechanism to abstract its property.
The insight of $\dphero$ is to leverage the knowledge encoded in the previously trained model to guide the subsequent allocation of noise heterogeneity, thereby leveraging the statistical perturbation and achieving enhanced utility.
Atop $\dphero$, we instantiate a heterogeneous version of DP-SGD and further extend it to federated training.
We conduct comprehensive experiments to verify and explain the effectiveness of the proposed $\dphero$, showing improved training accuracy compared with state-of-the-art works.
Broadly, we shed light on improving the privacy-utility space by learning the noise guidance from the pre-existing leaked knowledge encoded in the previously trained model, showing a different perspective of understanding the utility-improved DP training.
\end{abstract}



\section{Introduction}
\label{intro}

Deep learning has achieved remarkable success across a wide spectrum of domains~\cite{cvpr/LuLZL022,iclr/ZhangGL000S022,viswanathan2023prompt2model,zhao2024self}, primarily relying on the massive data used for model training.
As training data has been thoroughly analyzed to optimize model performance, a significant privacy concern arises regarding the model's potential to memorize individual data points~\cite{nips/BidermanPSSAPR23, uss/CarliniHNJSTBIW23,sp/LukasSSTWB23,acsac/0021ZCPTLY23}. 
Indeed, a growing body of studies~\cite{sp/ShokriSSS17,ccs/HitajAP17,uss/CarliniTWJHLRBS21} have demonstrated that it is feasible to identify the presence of a particular record or verbatim texts  in the training dataset, thereby raising severe privacy concerns.

Differential privacy (DP)~\cite{tcc/DworkMNS06,ccs/AbadiCGMMT016,iclr/PapernotSMRTE18,sp/Yu0PGT19}, emerged as de facto protection, can provide provable security for individuals’ privacy by adding the i.i.d noise to the sensitive data or computations.  
In detail, DP guarantees statistical indistinguishability between the outputs of two random mechanisms, which originate from the private inputs with or without a substituted individual data point. 
To protect sensitive data used in the training process, differentially private stochastic gradient descent (DP-SGD)~\cite{ccs/AbadiCGMMT016} has been proposed and regarded as a main-steam method.
The idea of DP-SGD is to add the homogeneous noise sampled from a Gaussian distribution to the aggregated gradient derived from a batch of examples in every training iteration.
Accordingly, DP-SGD, serving as the most popular baseline, can thwart an adversary from attaining membership of a particular data record when the adversary dissects an established model.

Subsequently, researchers identified the inherent trade-off between privacy and utility introduced by DP-SGD.
It is a well-known challenge to achieve high model utility/performance given meaningful DP guarantees~\cite{nips/TangPSM23,kdd/LeeK18,ccs/MohammadyXHZ0PD20,sp/Yu0PGT19,tit/GengV16,osdi/LuoPTCGL21} since acquiring strong protection realized by a large noise scale generally leads to unavoidable utility loss and performance degrading.
For example, the number of DP-SGD training iterations may increase by $10\times$ towards a similar utility metric compared with the pure SGD.
Accordingly, a research line of works~\cite{kdd/LeeK18,ccs/MohammadyXHZ0PD20,sp/Yu0PGT19,tit/GengV16}  explored to acquire a better utility by flexibly and empirically calibrate privacy budget allocation.
Regarding composition theorem, they try to either reallocate/optimize the privacy budget~\cite{kdd/LeeK18,sp/Yu0PGT19,tit/GengV16,osdi/LuoPTCGL21,nips/YangZZ0PL023} or modify the clip-norms~\cite{corr/abs-1908-07643,corr/abs-1812-02890}  of a (set of) fixed noise distribution(s) in each iteration.
These dynamic noise allocation solutions optimize the noise allocation in the range of the whole training process with a preset budget, but employs homogeneous noise at each iteration to perturb gradient updates.

Upon studying the iteration-wise utility with/without DP noise in the process of model convergence, we observe that  {utility loss can be ascribed to the homogeneity of noise applied to gradients.}
Regardless of the diverse features learned from the training data, homogeneous noise negatively contributes to training performance (\eg, convergence ability and accuracy) due to perturbing the original gradients derived in the backpropagation.
Drawing inspiration for dynamic noise allocation approaches, we believe introducing a noise heterogeneity view to the dynamic noise allocation approach will shed light on improving the privacy-utility space. 
Thus, we raise a fundamental question,

\emph{How do we improve the privacy-utility trade-off of DP-SGD by introducing the heterogeneous noise?}

\subsection{Technical Overview}

We consider a novel route of crafting iteration-wise noise heterogeneity by making use of pre-existing knowledge contained in the neural network, which captures the feature attributes prior-learned from the training data, thus improving the utility of the established model at every iteration.
The intuition is to dynamically allocate noise heterogeneity to diverse features in the back-propagation of SGD, in which the noise heterogeneity is guided by the prior learned knowledge contained in the existing model.
To this end, we propose a new framework -- differential privacy with heterogeneous noise ($\dphero$), guided by an iteration-wise guidance matrix derived from prior learned model parameters, to perturb the gradients derived in the backpropagation.
Specifically, we raise the following contributions progressively.

\smallskip
 \textbf{1) Allocating noise heterogeneity via pre-existing knowledge.}
To generate the model-guided heterogeneity, we propose a novel dynamic noise allocation scheme, where an iteration-wise (for short, stateful) matrix $\svec^{(t-1)}$ is computed using the pre-existing model at $(t-1)$-th iteration. With the notion of stateful $\svec^{(t-1)}$, we can guide the noise heterogeneity at the $t$-th training iteration. 
Namely, the stateful $\svec$ adjusts the noise $n$ used to perturb gradient updates at every iteration according to the heterogeneity derived by the $\svec^{(t-1)}$.
Consequently, the posterior random mechanism is guided by pre-existing knowledge encoded in prior model parameters at every training iteration.
Specifically, we consider a random mechanism named as $\dphero$ that adds  heterogeneous noise $\svec^{(t-1)} \cdot\mathbf{n}$ to gradients $\mathbf{g}^{(t)}$: 
$\mecha^{(t)}=\mathbf{g}^{(t)} +\svec^{(t-1)} \cdot\mathbf{n}$, where the abstraction of $\svec^{(t-1)}$ is independent to knowledge extraction function $\mathcal{F}$ of learned model $\nnet$ and indexed by states $t-1,t$.

For theoretical analysis, we abstract the notion of heterogeneous DP learning with stateful guidance for allocating noise heterogeneity.
By adopting composition~\cite{tcc/BunS16,csfw/Mironov17} and R\'enyi Divergence, we provide theoretical analysis on $\dphero$ SGD training following conventional proof style.
Accordingly, the instantiation of $\dphero$ SGD, regarded as a modified variant of standard DP-SGD, attains the standard DP guarantee.

\smallskip
\textbf{2) Constructing heterogeneous DP-SGD.}
We instantiate $\dphero$ as a heterogeneous version of DP-SGD, where the noise injected into gradient updates is heterogeneous. Specifically, the stateful $\svec^{(t-1)}$ at the $(t-1)$-th training iteration is derived from decomposition on model parameters $\mod^{(t-1)}$ at the prior training iteration, capturing the pre-existing knowledge.
Knowledge involved in $\mod^{(t-1)}$, serving as allocation guidance, propagates to the DP noise injected to gradients at the $t$-th training iteration, following the style of DP-SGD.
Accordingly, it captures the pre-existing statistical knowledge of private training data,  extracting heterogeneity applied to features. 
Later, the stateful guidance matrix $\svec^{(t-1)}$ adjusts the parameters of Gaussian distribution, equivalently affecting the heterogeneity of noise added to diverse features in the back-propagation of SGD.
Prior knowledge from extracted features has been reasonably DP-protected, thus not incurring extra risks in exposing private data.
The plug-in $\dphero$ SGD is generic and independent of learning models, best for both worlds in performance and privacy.
To demonstrate the generalization of $\dphero$, we present an extension of FedFed~\cite{nips/YangZZ0PL023} by replacing its SGD optimizer with $\dphero$ SGD.
This modification enables FedFed to achieve rigorous DP guarantees for client data while retaining its scalability and communication efficiency.

For test accuracy, $\dphero$ improves a series of state-of-the-arts, notably, from $95\%$ to $98\%$ over standard DP-SGD.
For training over the CIFAR-10 dataset, $\dphero$ improves by $18\%$-$47\%$.
We tested the convergence stability when adding small and large, showing that $\dphero$ could mitigate model collapse. 
At last, we visualize the DP-protected features during the training to explain $\dphero$'s superior performance.

\subsection{Contribution Summary} 

Overall, our contributions are summarized as follows.
\begin{enumerate}
	\setlength{\itemsep}{0pt}
	\item To form a step forward, we explore the relationship between DP training performance and heterogeneity at an iteration.
	Accordingly, we shed new light on bridging model utility and DP heterogeneity allocation to enhance the performance-privacy space.
    \item We propose a framework -- $\dphero$, supporting utility-improved training at every iteration by applying heterogeneous noise to model updates in back-propagation. 
	We abstract a guidance vector derived from pre-existing knowledge learned by models to guide noise heterogeneity applied to model back-propagation.
	Then, we formalize $\dphero$ and then provide theoretical analyses and proofs.
    \item To apply $\dphero$ SGD, we present an extension of FedFed~\cite{nips/YangZZ0PL023} with formal DP guarantee in federated learning.
	\item  Our $\dphero$ SGD is general and efficient, which could be adopted as a plug-in module. 
	$\dphero$ SGD could converge in fewer training iterations and mitigate the utility loss of the established model without relying on extra manual efforts. 
	Experiments and explanations confirm the superior improved privacy-utility trade-off. 
\end{enumerate}


\section{Preliminary}
\label{sec::prelim}
\subsection{General Notion of Differential Privacy}
Differential privacy (DP)~\cite{tcc/DworkMNS06,fttcs/DworkR14} theoretically guarantees individual privacy that the algorithm’s output changes insignificantly (see Definition~\ref{sensi}) if the inputting data has small perturbations.
Pure $\epsilon$-differential privacy is difficult to achieve in realistic learning settings, whereas the seminal work~\cite{ccs/AbadiCGMMT016}  training with SGD adopts approximate ($\epsilon, \delta$)-differential privacy, formally defined below.

\begin{definition}[Differential Privacy]
\label{def::dp}
	 A randomized mechanism $\mecha$ provides $(\epsilon,\delta)$-differential privacy if for any two neighboring datasets $D$ and $D'$ that differ in
	a single entry, $\forall S \subseteq Range(\mecha)$,
		\begin{equation}
		\mathrm{Pr}(\mecha(D)\in S)\leq e^\epsilon \cdot\mathrm{Pr}(\mecha(D')\in S)+\delta
			\end{equation}
	where $\epsilon$ is the privacy budget and $\delta$ is the failure probability. 
\end{definition}
\begin{definition}[Sensitivity] 
	\label{sensi}
	The sensitivity of a query function $\mathcal{F}:\mathbb{D}\rightarrow\mathbb{R}$ for any two neighboring datasets $D,D'$  is,
	 \begin{equation}
			\Delta = \max_{D,D'}\|\mathcal{F}(D)-\mathcal{F}(D')\|,
		 \end{equation}
	where $\|\cdot\|$ denotes $L_1$ or $L_2$ norm.
\end{definition}

Next, we introduce the definition of privacy loss~\cite{fttcs/DworkR14} on an outcome $o$ as a random variable when DP operates on two adjacent databases $D$ and $D'$.
Privacy loss is a random variable that accumulates the random noise added to the algorithm/model.
\begin{definition}[Privacy Loss~\cite{fttcs/DworkR14}]
	\label{def::priloss}
	Let $\mecha: \mathbb{D} \rightarrow \mathbb{R}$ be a randomized mechanism with input domain $D$ and range $R$.
	Let $D,D'$ be a pair of adjacent datasets and $\aux$ be an auxiliary input. For an outcome $o\in\mathbb{R}$, the privacy loss at $o$ is defined by,
	\begin{equation}
\mathcal{L}_{\mathsf{Pri}}^{(o)}\triangleq\log\frac{{\rm Pr}[\mecha(\aux,D)=o]}{{\rm Pr}[\mecha(\aux,D')=o]}
	\end{equation}
\end{definition}
where $\mathcal{L}_{\mathsf{Pri}}$ is a random variable on $r(o;\mecha,\aux,D,D')$, \ie, the random variable defined by evaluating the privacy loss at an outcome sampled from $\mecha(D)$.
Here, the output of the previous mechanisms is the auxiliary input $\aux$ of the mechanism $\mecha^{(t)}$ at $t$.

\subsection{DP Stochastic Gradient Descent}
DP-SGD~\cite{ccs/AbadiCGMMT016}, regarded as a landmark work, has been proposed to safeguard example-level model knowledge encoded from the training data, constrained by the privacy budget allocated at each training iteration.
As reported by DP-SGD, adding i.i.d. noise inevitably brings parameter perturbation over the learned model in practice. 
Research efforts such as~\cite{nips/TangPSM23, kdd/LeeK18, ccs/MohammadyXHZ0PD20, sp/Yu0PGT19, tit/GengV16, osdi/LuoPTCGL21} are focused on developing techniques that can provide stronger privacy guarantees while minimizing the loss of utility from various perspectives, \eg, clipping value optimization and privacy budget crafting. 
Zhou \etal~\cite{iclr/ZhouW021} improves utility by projecting the noisy gradients to a low-dimensional subspace, while $\dphero$ explores reusing the DP protected statistical knowledge learned from the private data. 

In DP learning, neighboring datasets $D,D'$ represent two datasets that
only differ by one training data point, while the $\mecha$ is the DP training algorithm.
Following the formality of the definition, the $\epsilon$ is an upper bound on
the loss of privacy, and the $\delta$ is the probability of breaking the privacy guarantee. 
DP-SGD is a differentially private version of stochastic gradient descent (SGD).
This approach adds noise to SGD computation during training to protect private training data.
The first step is to minimize the empirical loss function $\mathcal{L}(\theta)$ parameterized by $\theta$.
Secondly, gradient $\nabla_{\theta}\mathcal{L}(\theta,x_i)$ is computed at each step of the SGD, given a random subset of data $\bf x_i$.
The noise is added to the average gradients of $\nabla_{\theta}\mathcal{L}(\theta,x_i),\forall x_i$.
After training, the resulting model accumulates differentially private noise of each iteration to protect private individual data.

Through revisiting DP-SGD, we explore explaining the root
 of utility loss and bridge the concept of model-knowledge guidance and DP, making a DP training process fitting to enhance privacy-utility trade-offs better.
We showcase new thinking -- not employing auxiliary (e.g., public data) assistance for the higher model utility, and thus rethinking tolerant leakage (statistical knowledge, not membership, aligning standard DP definition) encoded in the prior DP-trained model.

\subsection{R\'enyi Differential Privacy}
R\'enyi differential privacy~\cite{csfw/Mironov17} has been proposed as a natural relaxation of differential privacy, particularly suitable for composing privacy guarantee of heterogeneous random mechanisms derived from algorithms. 
zCDP~\cite{tcc/BunS16} and R\'enyi DP~\cite{csfw/Mironov17} (RDP) are defined through R\'enyi Divergence by Bun~\etal~\cite{tcc/BunS16} for a tight analysis, thereby providing accumulating cumulative loss accurately and strong privacy guarantees.
Definition~\ref{def::renyidiv} presents the R\'enyi Divergence~\cite{csfw/Mironov17} for defining the R\'enyi differential privacy~\cite{csfw/Mironov17} as Definition~\ref{def:rdp}.

\begin{definition}[R\'enyi Divergence~\cite{csfw/Mironov17}]
	\label{def::renyidiv}
	For two probability distributions $P$ and $Q$  over $\mathbb{R}$,  R\'enyi divergence of order $\alpha$ is 
		\begin{equation}
			\mathcal{D}\alpha(P\|Q)\triangleq \dfrac{1}{\alpha-1}\log \mathbb{E}_{x\sim Q}[(\frac{P(x)}{Q(x)})^\alpha]
			\end{equation}
\end{definition}

Compared to standard differential privacy, R\'enyi differential privacy is more robust in offering an operationally convenient and quantitatively accurate way of tracking cumulative privacy loss throughout the execution of a standalone differentially private mechanism, such as iterative DP-SGD.
It supports combining the intuitive and appealing concept of a privacy budget by applying advanced composition theorems for a tighter analysis.
In return, an $(\alpha,\epsilon)$-R\'enyi DP implies $(\epsilon_\delta,\delta)$-DP for any given probability $\delta>0$ as Theorem~\ref{th:rdp-dp}.
We adopt the aforementioned DP advances to formalize DP with heterogeneous noise, devise the heterogeneous noise version of DP-SGD, and develop corresponding theoretical analyses.
\begin{definition}[R\'enyi Differential Privacy~\cite{csfw/Mironov17}]
\label{def:rdp}
    A randomized mechanism $\mecha \mathcal{D}:\mapsto\mathcal{R}$ is said to have $\epsilon$-R\'enyi differential privacy (RDP) of order $\alpha$ or $(\alpha,\epsilon)$-RDP for short, if for any adjcent $D,D'$,  R\'enyi divergence of random mechanisms satisfies that,
    \begin{equation}
        \mathcal{D}_\alpha(\mecha(D)|\mecha(D')\leq \epsilon
    \end{equation}
\end{definition}

\begin{theorem}[From RDP to $(\epsilon,\delta)$-DP~\cite{csfw/Mironov17}]
\label{th:rdp-dp}
    If $\mecha$ is an $(\alpha,\epsilon)$-RDP mechanism, it also satisfies $(\epsilon+\frac{\log 1/\delta}{\alpha-1},\delta)$-DP for any $0<\delta<1$.
\end{theorem}

\subsection{Security Model for Centralized DP Learning}
As for the security model, we consider a typical client-server paradigm of DP training. 
 The client, owning a set of private training data, trains a model conditioned on her private data, while the server receives the established model that is well-trained by the client, \ie, in a black-box manner.
The client trains a model conditioned on her data and sends the resulting model only to a remote server.
Assume a server is a malicious adversary, observes the final model, and tries to learn the existence of individual data.
Regarding Definition~\ref{def::dp}, privacy guarantee means resisting the server's inference on a particular record by viewing a differentially private model.
Our security model follows the privacy goal and adversary abilities that are the same as existing works since knowledge extraction is from the protected features on the client side.
$\dphero$ does not break existing settings or use any auxiliary data, thus incurring no extra privacy leakages to the server.


\section{Noise Heterogeneity in DP}
\label{sec::3}
To explore the noise heterogeneity, we start by adjusting the noise scale added to different elements, followed by witnessing the training process. 
Through repeated attempts, we observe that noise heterogeneity, \ie, the diverse noise scales added to the elements, can affect the training performance.
Accordingly, our idea is that prior model parameters (involving extracted elements with traditional DP protection) can guide the posterior random mechanism to improve training performance.
In the meantime, no privacy-sensitive element beyond DP protection is involved in yielding guidance.
Unlike dynamic allocation, we offer distinctive element-wise noise at each training step rather than scaling noise in a whole training process.

\subsection{Heterogeneity Guidance}
\label{sec::hete_guide}
Adding homogeneous DP noise uniformly across all model parameters often leads to suboptimal utility due to a mismatch between the injected noise magnitude and the underlying parameter sensitivity. 
Model parameters typically exhibit substantial heterogeneity in their magnitudes, gradient sensitivities, and contributions to overall task loss; applying the identical noise scale irrespective of these factors disproportionately corrupts low-norm or highly sensitive parameters, severely reducing their effective signal-to-noise ratio.
Conversely, larger or less-sensitive parameters can tolerate greater perturbations without substantial impact on convergence. 
Heterogeneous noise schemes address this issue by varying the injected noise based on structural properties such as parameter norm scaling, per-layer sensitivity, or gradient clipping statistics.
By preserving a more uniform effective SNR across parameters, heterogeneous noise enables models to maintain higher expressive power and faster convergence under the same global DP guarantee, yielding significantly improved privacy-utility tradeoffs.

We revisit reasonable leakages in DP models and make use of the pre-existing knowledge learned in the current model parameters to improve subsequent DP training performance.
Model training starts with an initial random model $\nnet^{(0)}$ towards a convergent model $\nnet^{(T)}$, which captures knowledge learned from data iteration by iteration.
Naturally, our idea is to introduce a scalar vector $\svec$ that is derived from the learned knowledge in $\nnet$ in the prior training process to serve as the guidance for subsequent DP training.

\subsection{Define Heterogeneous DP Learning}
\label{sec::def_hrm}
Consider a function $\mathcal{F}$ to denote functionality achieved by neural networks $\nnet$.
The $\nnet^{(t)}$, trained with the  DP mechanism, denotes the deep learning model at iteration $t$.
We formulate DP trained model at $t$-th iteration to be $\mathcal{F}(\nnet^{(t)}, \mathsf{Data})$ given private $\mathsf{Data}$. 
We utilize the $\svec^{(t-1)}$ at $t-1$-th iteration to adjust the next-step noise allocation at $t$-th iteration, where $\svec^{(t-1)}$ is computed by the prior $\nnet^{(t-1)}$ at $(t-1)$-th iteration  involving features learned in the first $t-1$ iterations.
Concretely, Definition~\ref{def::arm} introduces a general  notion of heterogeneous DP learning ($\dphero$) that varies with the time $t$, essentially adjusting the noise vector (sampled from Gaussian distribution) operated over the learning model.

\begin{definition}[Heterogeneous DP Learning]
	\label{def::arm}
	Let any two neighboring datasets $\mathsf{Data}$ and $\mathsf{Data}'$  differ in a single entry, $\epsilon$ be privacy budget, and $\delta$ be failure probability.
	Let $\noi$ be Gaussian noise distribution and $\mathsf{Data}$ be inputting private data.
	A time-dependent random mechanism of learning algorithm(s) $\mathcal{F}$  at the time $t$ is,
	 	\begin{equation}
  \mecha^{(t)}=\mathcal{F}(\nnet^{(t)}, \mathsf{Data}) +\svec^{(t-1)} \cdot \noi({\mu},{\sigma}^2)
		 	\end{equation}
\end{definition}
$\noi(\mu,\sigma^2)$ represents noise distribution with parameters $\mu, \sigma$.
To generate pre-existing knowledge stored in the model parameters, we can employ a knowledge-extraction method (\eg, principal component analysis~\cite{jolliffe2016principal}) to extract pre-existing knowledge stored in the learned model, saying
        $\svec^{(t-1)}\propto \mathcal{F}_{\mathsf{know}}(\nnet^{(t-1)}), t\in[0,T]$.
Accordingly, the noise sampled from the Gaussian distribution is scaled by $\svec$ (\ie, values and noise direction).
The $\svec$ keeps varied for tracking DP model training, calibrating noise vector via pre-existing knowledge stored in the model.
In summary, the $\dphero$ expects to: 
1)  be tailored with heterogeneous DP noise that is added to the learning process;
2) be generic and irrelevant to the convergence route for distinctive models for iteratively reaching a model optimum;
3) have good model accuracy and convergence performance given a preset privacy budget.

Intuitively, iteration-wise guidance enables utility-optimized training in every backpropagation.
Dynamic privacy-budget allocation assumes a constant budget in the whole training process, whereas  $\dphero$ assumes a pre-allocated budget in each iteration for acquiring relatively fine-wise optimization.
We consider $(t,\epsilon_t)$-utility-optimized DP in Definition~\ref{def::steputility} to capture the desirable property in DP learning.

\begin{definition}[$(t,\epsilon_t)$-Utility-Optimized DP]
	\label{def::steputility}
	Let any two neighboring datasets $D$ and $D'$  differ in a single entry, $\epsilon$ be privacy budget, and $\delta$ be failure probability.
   A mechanism $\dphero$ satisfies the following conditions at any training-iteration $t$:
\begin{itemize}
	\setlength{\itemsep}{0pt}
	\item[i] $($Privacy$)$. If for any two neighboring datasets ${D}$ and $ {D}'$, $\mathrm{Pr}(\mecha^{(t)}( {D})\in S)\leq e^{\epsilon_t} \cdot\mathrm{Pr}(\mecha^{(t)}( {D}')\in S)+\delta_t$ for any iteration $t\in [0,T]$.
	\item[ii] $($Utility$)$. Supposing an optimal $Z^\ast$, the objective function satisfies $\arg\  \min_{\epsilon_t= \epsilon/T} \mathcal{F}_{\mathsf{Diff}}[\mecha^{(t)}\|Z^\ast  ]$.
   \item[iii] $(t$-Sequential Composition$)$. If $\tilde{\mecha}=(\mecha^{(0)}, \ldots,\mecha^{(t)},\allowbreak\ldots,\mecha^{(T)})$, $\tilde{\mecha}$ satisfies $(\tilde{\epsilon}, \delta)$-DP such that $\tilde{\epsilon} \leq \epsilon$.
\end{itemize}
\end{definition}

Property (i) essentially guarantees differential privacy~\cite{tcc/DworkMNS06,fttcs/DworkR14} at each training iteration.
Property (ii) extracts the iteration-wise optimization, which expects that the difference measurement $\mathcal{F}_{\mathsf{Diff}}$ between the noisy model and pure model are small as possible.
In other words, at each training iteration, the algorithm ensure $(\tilde{\epsilon}, \delta)$-DP, while simultaneously seeking to minimize the divergence from an ideal output under a constrained privacy budget
Given a fixed privacy budget $\epsilon/T$, improving utility expects to reduce the difference between $\mecha^{(t)}$ and non-noisy $Z^\ast$.
Property (iii) makes sure that no additional privacy leakages are incurred in $\dphero$ under privacy composition, which is the same as the standard DP guarantee. 
Overall, Definition~\ref{def::steputility} formalizes a utility-optimized perspective on differential privacy by requiring mechanisms to preserve per-iteration privacy guarantees while explicitly minimizing a utility loss objective relative to an optimal target. 

\subsection{Overview of DP Heterogeneous SGD}

Before constructing DP heterogeneous SGD ($\dphero$ SGD), we adopt the notations of DP-SGD by revisiting standard DP-SGD~\cite{ccs/AbadiCGMMT016}.
DP-SGD trains a model with parameters $\mod$ by minimizing the empirical loss function $\mathcal{L}(\mod)$.
For a random example $x_i$, DP-SGD computes gradients $\mathbf{g}(x_i)\leftarrow \nabla(\mod, x_i)$ with clipping value $C$, and adds noise to $\sum_i\mathbf{g}(x_i)$ sampled from Gaussian distribution $\mathcal{N}(0, \sigma^2C^2\mathbf{I})$.
An adversary cannot view the training process except for the DP-trained model.

		\begin{algorithm}[H]
			\caption{$\dphero$ SGD}
			\label{alg::srnsgd}
			\begin{algorithmic}[1]
				\FOR{each training-iteration $t$}
				\STATE Sample a batch of data $\bf x$. 
				\FOR{$x_i\in {\bf x}$}
                    \STATE $\mathbf{g}_t\leftarrow\nabla_{\mod^{(t-1)}}\mathcal{L}(\mod^{(t-1)},x_i)$. \ENDFOR
				\STATE	 $\bar{\mathbf{g}}_t=\mathbf{g}_t/ \max(1, {\|\mathbf{g}_t\|}/{\mathsf{Cp}})$.
				\STATE $\svec^{(t-1)}$ computed by $\mod^{(t-1)}$ (such as Algorithm~\ref{alg::srn})
				\STATE $\tilde{\mathbf{g}_t} \leftarrow(\sum \bar{\mathbf{g}}_t +\svec^{(t-1)}\cdot\noi(\mu, \sigma^2\cdot\bm{I}))/S_{\bf x}$.
				\STATE $\mod^{(t)}\leftarrow \mod^{(t-1)}-\eta\tilde{\mathbf{g}_t}$.
          \ENDFOR
			\end{algorithmic}
		\end{algorithm}

Motivated by DP-SGD, we explore an instantiation of $\dphero$ to generate heterogeneous noise and then add a ``wisdom'' (guided by prior learned knowledge) heterogeneous noise.
Accordingly, we instantiate DP-SGD~\cite{ccs/AbadiCGMMT016} as the basis and replace its i.i.d. noise with heterogeneous noise.
In DP-SGD, the standard deviation $\sigma$ of $\noi(0,\sigma^2)$ is constant for each layer; however, our mechanism guided by $\svec$ adds different noise vectors for model updates at each iteration. 
With $\svec^{(t-1)}$, the added noise to each layer is guided by the learned model in the aspects of scales and noise space at every iteration.

Using $\dphero$ SGD, we implement an instantiated scheme of training a model starting from random initialization.
The first step is generating heterogeneous noise building on the covariance matrix of the model.
By principal component analysis (PCA)~\cite{jolliffe2016principal}, the noise matrix is tuned via the covariance matrix, which aligns with the subspace in which features exist.
PCA provides a natural mechanism to uncover and exploit heterogeneity in the underlying data distribution. Specifically, PCA rotates the model parameters into a new orthogonal basis where each principal component corresponds to an axis of maximal variance, and the associated eigenvalues quantify the variance along each direction. This decomposition reveals strong anisotropy: some directions exhibit substantially higher variability than others. In the context of DP learning, such heterogeneity is critical, as the sensitivity of the data to perturbations varies across directions. Allocating homogeneous noise across all dimensions fails to respect this structure, disproportionately affecting components with low intrinsic variance. By leveraging PCA, noise can be modulated according to the variance structure, \ie,    injecting smaller noise in high-variance directions that are more robust to perturbations, and simultaneously allocating more regularization or larger noise to fragile, low-variance directions. 

Formally, by transforming the model parameters into the PCA basis, applying direction-dependent noise proportional to the inverse eigenvalues of the covariance matrix, and transforming back, one achieves a heterogeneous noise model that preserves critical information while satisfying DP constraints. Thus, PCA serves as a powerful tool for constructing heterogeneity-aware noise mechanisms, enabling significantly improved privacy-utility trade-offs compared to uniform noise baselines.
When training with SGD, updatable gradients computed in the backpropagation are added by noise, whose scales are guided by the subspace generated by PCA.
We consider extracting pre-existing knowledge from whole model parameters rather than a layer to capture the whole statistical space.
In this way, the noise space is more comprehensive, and the noise scale is more adaptive to the feature space.

\subsection{Detailed Construction}
\subsubsection{Construction of  $\dphero$ SGD} 
Regarding Definition~\ref{def::steputility}, achieving strong utility under tight privacy constraints demands carefully shaped noise: uniform noise across dimensions can disproportionately damage sensitive components, leading to unnecessary degradation in optimization performance. This motivates the use of PCA, which reveals intrinsic heterogeneity in the update space by decomposing it into orthogonal directions with varying variance. 
PCA identifies principal directions where the model is naturally more robust to perturbations, enabling noise allocation that is inversely aligned with direction-specific sensitivity. 
By injecting smaller noise along high-variance directions and larger noise along low-variance ones, a PCA-based noise mechanism better preserves informative structures in the model updates while satisfying the same overall DP guarantees. 
Thus, PCA provides a principled and structure-aware strategy to optimize the critical trade-off between privacy preservation and utility maximization articulated in Definition~\ref{def::steputility}.

\noindent\textbf{Step-1}. 
Assume that the model $\mod^{(0)}$ is initialized to be random during training.
The model parameters at each iteration $t$ represent the learning process of features in the dataset; \ie, the training is to optimize the model parameters by capturing data attributes.
The $\mod^{(t-1)}$ takes a set of inputting data ${\bf x}$ in size $S_{\bf x}$ (\ie, batch size) and compute the gradient
\begin{equation}
	\mathbf{g}_t\leftarrow\nabla_{\mod^{(t-1)}}\mathcal{L}(\mod^{(t-1)},x_i), x_i\in {\bf x}
\end{equation}
The $\mathbf{g}_t$ is clipped with the clip value $\mathsf{Cp}$, thus ensuring that the gradients are scaled to be of norm $\mathsf{Cp}$.
The clipped gradients are
	$ \bar{\mathbf{g}}_t$ handled with clip value $\mathsf{Cp}$.

\noindent\textbf{Step-2}. In our implementation,  $\svec^{(t-1)}$ can be realized by following Algorithm~\ref{alg::srn} using $\mod^{(t-1)}$.
Since $\svec^{(t-1)}$ is varied at each training iteration, $\svec^{(t-1)}$-guided noise distribution operating on gradients is varied during the whole training process.
$\svec^{(t-1)}$ contains the computed sub-space $\subsp^{(t-1)}$ and eigenvalues matrix $\vsvd^{(t-1)}$ extracted from prior-learned model.
From a practical view, $\subsp^{(t-1)}$ configures the direction of the noise to be added.
$\vsvd^{(t-1)}$ generated from singular value decomposition is utilized to scale the noise distribution.
Here, independent and identically distributed noise can be sampled from a standard noise distribution $\noi$, such as Gaussian and Laplace distributions. 
The generation of $\svec^{(t-1)}$ does not introduce extra leakage since $\mod^{(t-1)}$ learned in the prior $t-1$ iterations has been well-protected through $\dphero$ SGD.
 
\noindent\textbf{Step-3}. 
Following the logic of DP-SGD, $\svec^{(t-1)}$-guided noise is added to a batch of gradients,
\begin{equation}
	\begin{aligned}
		\tilde{\mathbf{g}_t} \leftarrow(\sum \bar{\mathbf{g}}_t +\svec^{(t-1)}\cdot\noi(\mu, \sigma^2\cdot\bm{I}))/S_{\bf x}
	\end{aligned}
\end{equation} 
$\vsvd^{(t-1)}$ here is different at every backpropagation of different layers, achieving different noise levels on each layer.
This layer-wise noise tuning speeds up the convergence and mitigates model collapse.
It derives from the corresponding model parameters of a unique layer that is relevant to an iteration $t$ at the current backpropagation.
$\dphero$ SGD is independent of the choices of optimizer and optimizers, which could be potentially generalized to different learning models without  much effort of manual tuning.

\noindent\textbf{Step-4}. The last step is to perform gradient decent $ \mod^{(t)}\leftarrow \mod^{(t-1)}-\eta\tilde{\mathbf{g}_t}$ using the new noisy gradients $\tilde{\mathbf{g}_t}$, where $\eta_t$ is a preset scalar.
For attaining higher utility, adding noise should avoid hurting important features (extracted by the model for later prediction.
Finally, the model converges better since the space of model parameters (regarded as a matrix) is relatively less destroyed by using the noise sampled from the identical space.

\subsubsection{Construction of Noise Guidance}
\label{subsec::workflow}

The math tool, principal component analysis (PCA)~\cite{nips/Shawe-TaylorW02} performs analyzing data represented by inter-correlated quantitative dependent variables.
 It forms a set of new orthogonal variables, called 
components, depending on the matrix eigen-decomposition and singular value decomposition (SVD).
Given a matrix $\mathbf{X}$, of column-wise mean equal to $0$, the multiplication $\mathbf{X}^\top \mathbf{X}$ is a correlation matrix.
Later, a diagonal matrix of the (non-zero) eigenvalues of $\mathbf{X}^\top \mathbf{X}$ is extracted together with the eigenvectors.
Essentially, PCA simplifies data representation and decomposes its corresponding structures.

We propose a simple yet efficient approach by examining the model parameters as a result of knowledge integration over diverse features extracted from private data. 
As in Algorithm~\ref{alg::srn}, we employ the PCA decomposition~\cite{jolliffe2016principal} to extract knowledge learned by the training model and apply generated guidance $\svec^{(t)}$ at iteration $t$ to adjust noise addition at the next iteration.
PCA decomposition can extract knowledge from representative data (\ie, model parameters in our setting) by analyzing inter-correlated quantitative dependence.
Normally, a neural network kernel extracting the features from the images is a matrix that moves over the input data to perform the dot product with the sub-region of input data.
Denote $\mathbb{R}$ to be the real number set. 
Let $\mathbf{b}=[b_1, b_2, \ldots, b_k]$ be a vector, and $\mathbf{B}=[\mathbf{b_1},\mathbf{b_2},\ldots, \mathbf{b_d}]^\top\in \mathbb{R}^{n\times m}$ be a matrix.

\begin{center}
   \begin{minipage}{0.96\linewidth}
		\begin{algorithm}[H]
			\caption{DP Heterogeneous Noise Guidance}
			\label{alg::srn}
			\begin{algorithmic}[1]
				\STATE Compute $\tilde{\mod}^{(t)} = \mod^{(t)}(\mod^{(t)})^\top$
				\STATE Compute $\vsvd^{(t)},\subsp^{(t)} \leftarrow \pca(\tilde{\mod}^{(t)})$
				\STATE Compute 	$\svec^{(t)} = \subsp^{(t)}\cdot \vsvd^{(t)}$
			\end{algorithmic}
		\end{algorithm}
	\end{minipage}
\end{center}

\noindent\textbf{Step-1}. For each layer, the client calculates $\mod^{(t)}(\mod^{(t)})^\top$ to attain $\tilde{\mod}^{(t)}\in\mathbb{R}^{k\times k}$.

\noindent\textbf{Step-2}. The client performs principle component analysis $\pca(\tilde{\mod}^{(t)})$ to give the sub-space $\subsp^{(t)}\in \mathbb{R}^{d\times k}$.
The algorithm $\pca$ reduces the dimensions and encodes $\mod^{(t)}$ into a compact representation that is good enough to analyze and represent current  $\mod^{(t)}$. Simultaneously, the client computes singular value decomposition $\dot{\vsvd^{(t)}}=\pca(\tilde{\mod}^{(t)})$ through PCA and transform $\pca(\tilde{\mod}^{(t)})$ to eigenvalues matrix $\vsvd^{(t)}\in \mathbb{R}^{k\times k}$ by $\dot{\vsvd^{(t)}}(\dot{\vsvd^{(t)}})^\top$.
The $\vsvd^{(t)}$ is employed as the scalar matrix to adjust noise scales for a batch of gradients in $t$-th training iteration. 

\noindent\textbf{Step-3}. $\svec^{(t)}$ is computed by multiplying $\vsvd^{(t)}$ and $\subsp^{(t)}$, which are further utilized to guide the noise added to gradients in every backpropagation. 

\subsubsection{Noise Guidance through Pre-existing Knowledge}
For a non-private model, $\mod$ converges to a stable status through uncountable routes of optimizing model parameters.
Noise addition becomes complicated if we refer to different optimization tools; it is no longer generic.
The addition of noise in $\mod$ inevitably has a negative contribution to the extraction of features from private data compared to pure parameters.

$\svec$ achieves improved allocation of parameter-wise heterogeneous noise at each training iteration with the constraint of a preset privacy budget.
This automatic allocation is generated from the prioritization evaluation of the training model in a differentially private manner. 
From this viewpoint,  injecting noise into the model parameters negatively contributes to both the knowledge and the process of knowledge integration. 
Compared with DP-SGD, the proposed method mitigates the destruction of the process of knowledge integration while keeping the learned knowledge unchanged. 
Different grid search for tuning hyperparameters, $\dphero$ SGD adjusts the intermediate training process via instantaneous learnable parameters rather than setting a set of possibilities. 
Combining grid search (vertically tuned) and $\dphero$ SGD (horizontally tuning) may further boost the automatic optimization of DP learning in an algorithmic view.

\subsection{Federated Training with $\dphero$}
FedFed~\cite{nips/YangZZ0PL023} is a hierarchical federated optimization protocol that divides clients into groups and employs a two-level communication structure. Clients locally compute model updates using stochastic gradient descent (SGD), and then communicate either intra-group or inter-group for aggregations, significantly reducing the global communication cost.
To extend FedFed with formal DP guarantees, we apply per-client $\dphero$ SGD in the local training phase. The modified procedure for each client in FedFed is as Algorithm~\ref{algo:FedFest_with_hero}.

The training process proceeds in rounds orchestrated by a central server. At the start, the server distributes the initial global model $\phi^0$ and a globally shared dataset $\mathcal{D}^s$ to all clients, who each combine it with their local private dataset $\mathcal{D}^k$ to form $\mathcal{D}_t^k = \mathcal{D}^k \cup \mathcal{D}^s$. In each communication round $r$, the server randomly selects a subset of clients $\mathcal{C}_r \subseteq {1, \ldots, K}$, transmits the current global model $\phi^r$, and waits while each selected client $k \in \mathcal{C}_r$ performs $E_r$ epochs of local training using the DP-Hero SGD optimizer on their combined dataset $\mathcal{D}_t^k$. This approach ensures local updates achieve differential privacy via gradient clipping and noise addition. Upon completion, clients send their updated model parameters $\phi_k^{r+1}$ to the server, which aggregates these updates (e.g., via $AGG({\phi_k^{r+1}})$) to form the next global model $\phi^{r+1}$. This iterative process allows the system to collaboratively train a privateglobal model with DP guarantee, while benefiting from the scalability and communication efficiency of the FedFed.

\begin{algorithm}[t]
	\caption{FedFed with $\dphero$ SGD}
	\label{algo:FedFest_with_hero}
	\textbf{Server Input: } initial global model $\phi^0$, communication round $ T_r  $. \\
	\textbf{Client $k$'s Input: } local epochs $ E_r$, local private datasets $\mathcal{D}^k$, learning rate $\eta_k$.
 
	\begin{algorithmic}
		\STATE {\bfseries Initialization:} server distributes the initial model $\phi^0$ to all clients, 
        \STATE  {Generate globally shared dataset $\mathcal{D}^s$.} 
        \STATE  {Distribute $\mathcal{D}^s$ to all clients and $\mathcal{D}_{t}^k = \mathcal{D}^k \cup \mathcal{D}^s$.}
        \STATE 
        \STATE \textbf{Server Executes:}
            \FOR{each round   $ r=1,2, \cdots,  T_r $}
        		\STATE server samples a subset of clients $C_r \subseteq \left \{1, ..., K \right \}$
                \STATE server \textbf{communicates} $\phi^r$ to selected clients $k \in C_r $
        		\FOR{ each client $ k \in C_{r}$ \textbf{ in parallel }}
                       \STATE $\small \phi_{k}^{r+1} \leftarrow \textbf{Client\_Training}(k, \phi^{r}$)
        		\ENDFOR
             \STATE $ \phi^{r+1} \leftarrow AGG(\phi_{k}^{r+1}) $
        	\ENDFOR
       \STATE
        \STATE \textbf{Client\_Training($ k, \phi^r$):}
        \STATE $\phi^r$ initialize local model $\phi_k^r$
            \FOR{each local epoch $e$ with $e=1,2,\cdots, E_r $}
             \STATE$\small \phi_{k}^{r+1} \leftarrow$ $\dphero$ SGD update with $\mathcal{D}_t^k$  
             \ENDFOR
        \STATE \textbf{Return} $\phi_{k}^{r+1} $ to server
\end{algorithmic}
\end{algorithm}


\subsection{Privacy Analysis and Theoretical Explanation}
\label{sec::privacy}

We establish the fundamental privacy guarantee for DP-Hero mechanisms. It shows that if the noise scale ($\sigma$) is chosen according to the lower bound established by ~\cite{ccs/AbadiCGMMT016}, then the $\dphero$ mechanism at each iteration attains ($\epsilon,\delta$)-DP. Theorem~\ref{dp_srnsgd} ensures that the per-iteration privacy risk of DP-Hero is comparable to that of classical DP mechanisms, provided the noise is correctly calibrated.

\begin{theorem}
	\label{dp_srnsgd}
	Let a random mechanism $\mecha^{(t)}$ be $(\epsilon^{(t)}, \delta)$-differential privacy at the iteration $t$.
	A $\dphero$ mechanism $\Tilde{\mecha}^{(t)}$ parameterized by $(\Tilde{\epsilon}^{(t)},\delta)$ is $({\epsilon}, \delta)$-differential privacy if $\Tilde{\sigma} =\sigma$, where  $\sigma \geq c_2q\sqrt{T\log(1/\delta)}/{\epsilon}$~\cite{ccs/AbadiCGMMT016}.
\end{theorem}

\begin{proof}
    Standard DP-SGD is $(\epsilon,\delta)$-differentially private  if $\sigma \geq c_2q\sqrt{T\log(1/\delta)}/{\epsilon}$ for any $\delta>0$~\cite{ccs/AbadiCGMMT016}.
The $q, T$ are, respectively, sampling probability and the number of steps relevant to model training.
The $c$ is a  constant for all DP mechanisms.
Take $\mecha^{({t})}$ to be a $\dphero$ random mechanism that is derived from  $(\epsilon, \delta)$-differential privacy.
The $\mecha_{\mathsf{t}}$ has the same configuration of $q, T,c$ due to the identical training procedure.
If $\sigma$ is unchanged, $\mecha^{({t})}$ also satisfies $\sigma \geq c_2{q\sqrt{T\log(1/\delta)}}/{\epsilon}$ for any $\delta>0$.
Thus, $\mecha_{\mathsf{t}}$  is $(\epsilon, \delta)$-differentially private.
\end{proof}

Then, we demonstrate  that by appropriately parameterizing the diagonal scaling matrix ($\svec$), the total variance of noise injected by $\dphero$ SGD can be made equal to that of standard DP-SGD. This means that the two mechanisms have matching privacy and utility properties from the perspective of total noise magnitude as in Theorem~\ref{the:dpsgd-para}.
\begin{theorem}
	\label{the:dpsgd-para}
	Let $\dphero$ SGD be parameterized by $\noi(0,\tilde{\sigma}^2)$ and standard DP-SGD be parameterized by $\noi(0,\sigma^2)$, repectively.
    Consider $\dphero$ SGD adds noise from  $\mathcal{N}(0,\Tilde{\sigma}^2)$ and standard DP-SGD uses $\mathcal{N}(0, \sigma^2)$.
    Let $\svec$ be a diagonal matrix whose $i$-th diagonal entry is defined as	$v_i={\egv_i\cdot\sqrt{k}\cdot\sigma}/{\sqrt{\sum_{i=1}^{k} {\egv_i}^2}}$, where $k$ denotes the dimension. Then, if $\tilde{\sigma} = \sigma$, the total variance of the noise injected by DP-Hero SGD with $\svec$ equals that of standard DP-SGD.
\end{theorem}

\begin{proof}
For generating noise, we need to keep $\tilde{\sigma}^2=\sigma^2$ to guarantee the same size of noise sampled from the distributions $\ournoi,\noi$.
Let $n$ sampled from Gaussian distribution be $\nos \leftarrow\noi(\mu,\sigma^2)$.
For sampling $k$ times (until iteration $k$) from Gaussian distribution, we have the expectation of $\noimatri$,
\begin{equation}
\mathbb{E}[\noimatri^\top\cdot \noimatri]= \mathbb{E}[\sum_{i=1}^{k} (n_i)^2]=t\sigma^2
\end{equation}
For sampling $k$ times from $\ournoi$, we require the following expectation to satisfy $\mathbb{E}[ (\subsp \vsvd\noimatri)^\top\cdot \subsp\vsvd\noimatri]=t\sigma^2$.
This equation gives the relation $\sum_{i=1}^{k} v_i^2=k\sigma^2$.
That is, a feasible solution of $v_i$ is set to be
$v_i={\egv_i\cdot\sqrt{k}\cdot\sigma}/{\sqrt{\sum_{i=1}^{k} {\egv_i}^2}}$.
\end{proof} 

Building on $\alpha$-R\'enyi divergence and privacy loss, concentrated differential privacy (CDP)~\cite{tcc/BunS16} allows improved computation mitigating single-query loss and high probability bounds for accurately analyzing the cumulative loss.
It centralizes privacy loss around zero, maintaining sub-Gaussian characteristics that make larger deviations from zero increasingly improbable.
In return, zero-CDP implies $(\epsilon_{\rho,\delta},\delta)$-DP as restated in Theorem~\ref{the:cdp-dp}~\cite{tcc/BunS16}.

\begin{definition}[zero-CDP~\cite{tcc/BunS16}]
	\label{def::RDP}
	A randomized mechanism $\mecha$ is said to be $\rho$ zero-concentrated differentially private if for any neighboring datasets $D$ and $D'$, and all $\alpha\in(1,\infty)$, we have,
		\begin{equation}
				\begin{aligned}
			\mathcal{D}_\alpha(\mecha(D)|\mecha(D'))=\dfrac{1}{\alpha-1}\log \mathbb{E}[e^{(\alpha-1)\mathcal{L}_{\mathsf{Pri}}^{(o)}}]\leq \rho \alpha
					\end{aligned}
			\end{equation}
	where $\mathcal{L}_{\mathsf{Pri}}^{(o)}$ is privacy loss and $\mathcal{D}_\alpha(\mecha(D)|\mecha(D'))$ is $\alpha$-R\'enyi divergence between the distributions of $\mecha(D)$ and $\mecha(D')$.
\end{definition}

\begin{theorem}[From zero-CDP to $(\epsilon,\delta)$-DP~\cite{tcc/BunS16}]
\label{the:cdp-dp}
If a random mechanism $\mecha$ is $\rho$-zero-CDP, then $\mecha$ also provides $(\rho+2\sqrt{\rho \log (1/\delta)},\delta)$-DP for any $\delta>0$.
\end{theorem}

At last, since we have aligned the privacy guarantee of $\dphero$ with the standard DP, we follow the standard composition-paradigm proof~\cite{csfw/Mironov17}  under the definition of zCDP~\cite{corr/DworkR16,tcc/BunS16,sp/Yu0PGT19} through R\'enyi Divergence by Bun~\etal~\cite{tcc/BunS16} for a tight analysis, resulting in Theorem~\ref{the::compose-dpsgd}.

 \begin{theorem}[Composition of $\dphero$ SGD] 
 \label{the::compose-dpsgd}
 Let a mechanism consist of $T$ $\dphero$ mechanisms: $\mecha=(\mecha^{(1)},\allowbreak\dots,\mecha^{(T)})$.
	Each $\dphero$ SGD $\mecha^{(t)}: \mathbb{D}^{(t)}\rightarrow\mathbb{R}$ satisfies $\rho^{(t)}$-zCDP, where the $\mathbb{D}^{(t)}$ is a subset of $\mathbb{D}$.
	The mechanism $\mecha$ satisfies $((\max_{t} \rho^{(t)})\allowbreak+2\sqrt{(\max_{t} \rho^{(t)})\log(1/\delta)},\delta)$-differential privacy.
\end{theorem}

\begin{proof}
    Consider two neighboring datasets $D,D'$.
By Theorem~\ref{the:dpsgd-para}, our mechanism at each iteration adds the noise equal to being sampled from $\noi(0,\sigma^2)$.
By Definition~\ref{def::RDP} and Definition~\ref{def::arm}, we calculate,
\begin{equation}
	\begin{aligned}
		&\sqrt{2 \pi\sigma^2 }\exp\left[(a-1)\mathcal{D}_\alpha(\mecha(D)|\mecha(D'))\right]\\
  &=\int_{\mathbb{R}}e^{\left({-\alpha(x-\mathcal{F}(D))^2}/{2\sigma^2}-{(1-\alpha)(x-\mathcal{F}(D'))^2}/{2\sigma^2}\right)}dx\\
		&=\int_{\mathbb{R}}e^{\left(-{(x-(\alpha\mathcal{F}(D)+(1-\alpha)\mathcal{F}(D')))^2}/{2\sigma^2}\right)}dx\\
 &\quad+\int_{\mathbb{R}}e^{\left({(\alpha\mathcal{F}(D)+(1-\alpha)\mathcal{F}(D'))^2-\alpha\mathcal{F}(D)^2}/{2\sigma^2}\right)}dx\\
		&\quad-\int_{\mathbb{R}}e^{\left({(1-\alpha)\mathcal{F}(D')^2}/{2\sigma^2}\right)}dx\\
    &=\sqrt{2 \pi\sigma^2 }\exp({\alpha(\alpha-1)(\mathcal{F}(D)-\mathcal{F}(D'))^2}/{(2\sigma^2)})\\
	\end{aligned}
\end{equation}
Thus,
\begin{equation}
	\begin{aligned}
				 &\exp\left[(a-1)\mathcal{D}_\alpha(\mecha(D)|\mecha(D'))\right]\\
		&=\exp({\alpha(\alpha-1)(\mathcal{F}(D)-\mathcal{F}(D'))^2}/(2\sigma^2))\\
		&=\exp({\alpha(\alpha-1)\Delta^2}/{(2\sigma^2)})
	\end{aligned}
\end{equation}
By the result ${\alpha(\alpha-1)\Delta^2}/{(2\sigma^2)}$, this calculation tells that our noise mechanism follows $(\Delta^2/2\sigma^2)$-zCDP at each iteration.
    
By Definition~\ref{def::priloss} and $\mathbb{E}\left[e^{(\alpha-1)\mathcal{L}_{\mathsf{Pri}}^{(o),(t)}}\right]$~\cite{sp/Yu0PGT19}, we have,
\begin{equation}
	\mathbb{E}\left[\left(\frac{{\rm Pr}(\mecha(\aux,D)=o)}{{\rm Pr}(\mecha(\aux,D')=o)}\right)^{\alpha-1}\right]	\leq e^{(\alpha-1)\alpha\cdot(\max_{t} \rho^{(t)})}
\end{equation}
By Markov’s inequality, calculate the probability,
\begin{equation}
	\begin{aligned}
		\mathrm{Pr}[\mathcal{L}_{\mathsf{Pri}}^{(O)}\geq \epsilon ]&=	\mathrm{Pr}[e^{(\alpha-1)\mathcal{L}_{\mathsf{Pri}}^{(O)}}>e^{(\alpha-1)\epsilon}]\\
		&\leq \frac{\mathbb{E}\left[e^{(\alpha-1)\mathcal{L}_{\mathsf{Pri}}^{(O)}}\right]}{e^{(\alpha-1)\epsilon}}\leq e^{(\alpha-1)(\alpha(\max_{t} \rho^{(t)})-\epsilon)}
	\end{aligned}
\end{equation}
Subject to $\sigma={\sqrt{2T\log{(1/\delta)}}}/{\epsilon}$, we use $\alpha = \frac{\epsilon+(\max_{t} \rho^{(t)})}{2\cdot(\max_{t} \rho^{(t)})}$ as derived in~\cite{sp/Yu0PGT19}, and compute,
\begin{equation}
	\begin{aligned}
\mathrm{Pr}[\mathcal{L}_{\mathsf{Pri}}^{(O)}> \epsilon ]\leq e^{-(\epsilon-(\max_{t} \rho^{(t)}))^2/(4\cdot(\max_{t} \rho^{(t)}))}\leq \delta
	\end{aligned}
\end{equation}
For any $S$ in Definition~\ref{def::dp},
\begin{equation}
	\begin{aligned}
		&\leq\mathrm{Pr}[O\in S\wedge \mathcal{L}_{\mathsf{Pri}}^{(O)} \leq \epsilon] +\mathrm{Pr}[\mathcal{L}_{\mathsf{Pri}}^{(O)}>\epsilon]\\
		&\leq\mathrm{Pr}[O\in S\wedge \mathcal{L}_{\mathsf{Pri}}^{(O)} \leq \epsilon] +\delta\\
		&\leq \int_o \mathrm{Pr}[\mecha(D')=o|o\in S]e^\epsilon do +\delta\\
		&=e^\epsilon\mathrm{Pr}[\mecha(D')=S]+\delta
	\end{aligned}
\end{equation}
still satisfies original DP definition, as  in~\cite{tcc/BunS16,csfw/Mironov17}.
\end{proof}

Together, these theorems ensure that DP-Hero SGD matches standard practice in DP-SGD both in per-step and composed privacy, and that the parameterization of $\dphero$ SGD can be mapped to the well-established privacy analysis in the previous work.

\subsection{Linear Layer Analysis as an Example}
We consider a binary classification for simplification and then instantiate a linear layer correlation analysis as an example supplement.
We regard SGD training as ``ground truth''. 
We simplify model parameters as an abstraction of extracted features over the whole dataset.
Define layer-wise model parameters to be $\mod$ in a binary classification model.
Let the $y\in\{-1,1\}$ be model output, $(x,y)$ be the input-output pair.
Let  noise overall features be $\noimatri$, where the norm $\|\noimatri \|$ maintains to be the same.
We expect the noise addition to not affect the space of model parameters and to keep the individual information in the model parameters unleaked.
Our objective is to minimize the variation of model outputs  from DP training and pure model at each training iteration, \ie,
\begin{equation}
	\begin{aligned}
		 \arg \min_{\|\noimatri\|}  |\sum_{i}\int(\mod+\noimatri)x_i{y_i}' \,dn -\sum_{i}\int \mod x_iy_i \,dn|
	\end{aligned}
\end{equation}
Consider that noise variable $n$ being injected into each feature could be continuous ideally. 
Since it is sampled from a distribution with a mean value of $0$, the integration of $n$ equals $0$, which could be removed for simplification.

We expect the first part to be large (denoting high utility) and the difference between the two parts to be as small as possible.
Then, we define the variance to be,
\begin{equation}
	\label{equa::vardiff}
	\var[\sum_{i}(\mod+\noimatri)x_i{y_i}'-\sum_{i}\mod x_iy_i]
\end{equation}
Equation~\ref{equa::vardiff} measures the difference of average correction of two models.
Equation~\ref{equa::vardiff} can be simplified by the expectation,
\begin{equation}
	\label{equa::varexp}
\mathbb{E}[\sum_{i}((\mod+\noimatri)x_i{y_i}'-\mod x_iy_i)]
\end{equation}
For linear transformation, we get,
\begin{equation}
\label{eqa::binary}
\begin{aligned}
    &(\mod+\noimatri)^\top x_i{y_i}'-\mod^\top x_iy_i \\
   &= (\mod+\noimatri)^\top x_i{(y_i+\Delta y_i)}-\mod^\top x_iy_i\\
		&= \mod^\top x_i{\Delta y_i}+\noimatri^\top x_i{y_i}+\noimatri^\top x_i{\Delta y_i} \\
		&=   (\mod+\noimatri)^\top x_i{\Delta y_i}+\noimatri^\top x_i{y_i}\\
		&= (\mod+\noimatri)^\top x_i\noimatri^\top x_i+\noimatri^\top x_i\mod^\top x_i\\
    & =(\mod^\top\noimatri^\top+\noimatri^\top\noimatri^\top+\noimatri^\top\mod^\top){x_i}^2
\end{aligned}
\end{equation}
Specifically, if $(\mod+\noimatri)^{\top}x_i$ is close to $y_i$, the differentially-private (noisy for short) model accuracy is high.
To attain the minimizer, we could solve Equation~\ref{eqa::binary} by $\mod \bot \noimatri$. 
In this example analysis, attaining support for the noise-model relation is enough for simplification.

\begin{figure*}
	\centering
 \begin{minipage}[t]{0.31\textwidth}
 \centering
		\includegraphics[width = 1.8in]{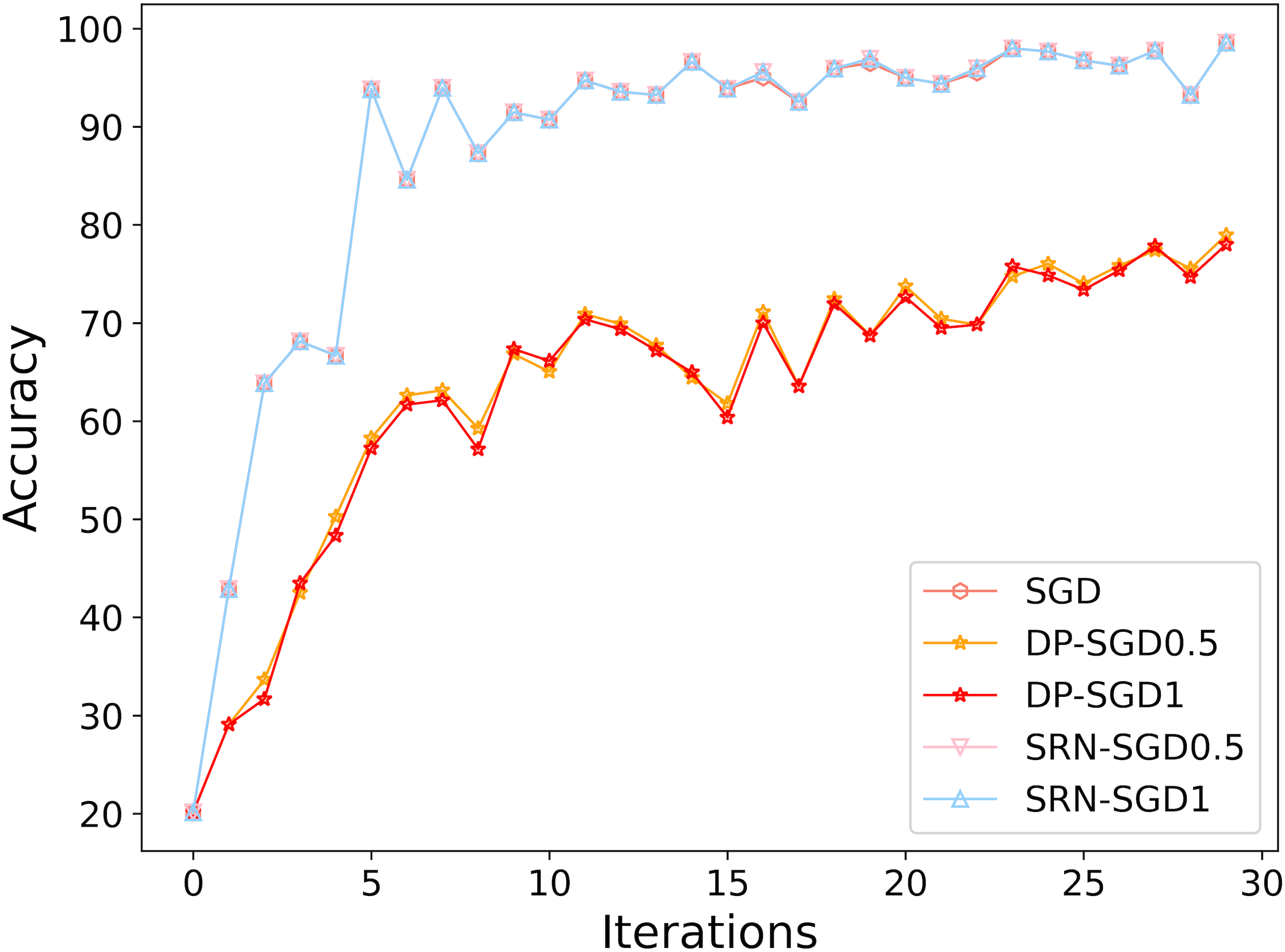}
		\subcaption{$\sigma = 0.5,1$}
		\label{sigma_0.5_1_iteration}
	\end{minipage}
	\begin{minipage}[t]{0.31\textwidth}
 \centering
		\includegraphics[width = 1.8in]{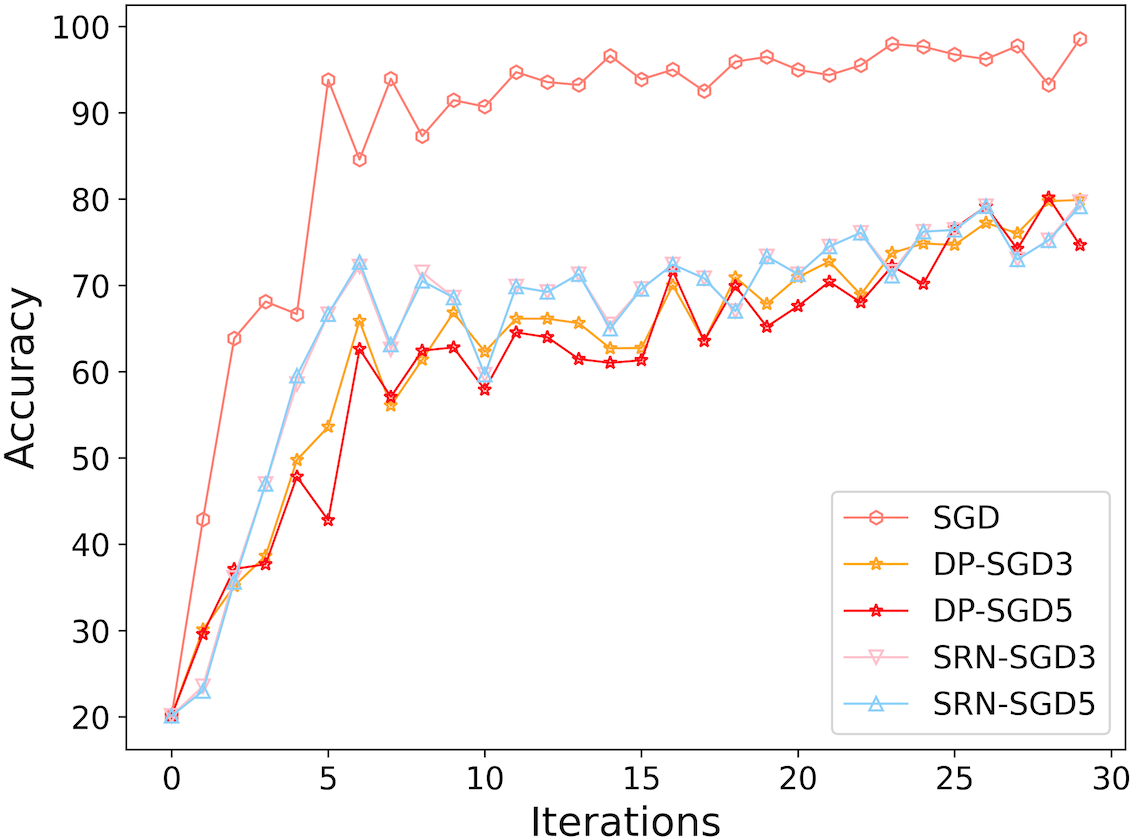}
		\subcaption{$\sigma = 3,5$}
		\label{sigma_3_5_iteration}
	\end{minipage}
	\begin{minipage}[t]{0.31\textwidth}
 \centering
		\includegraphics[width = 1.8in]{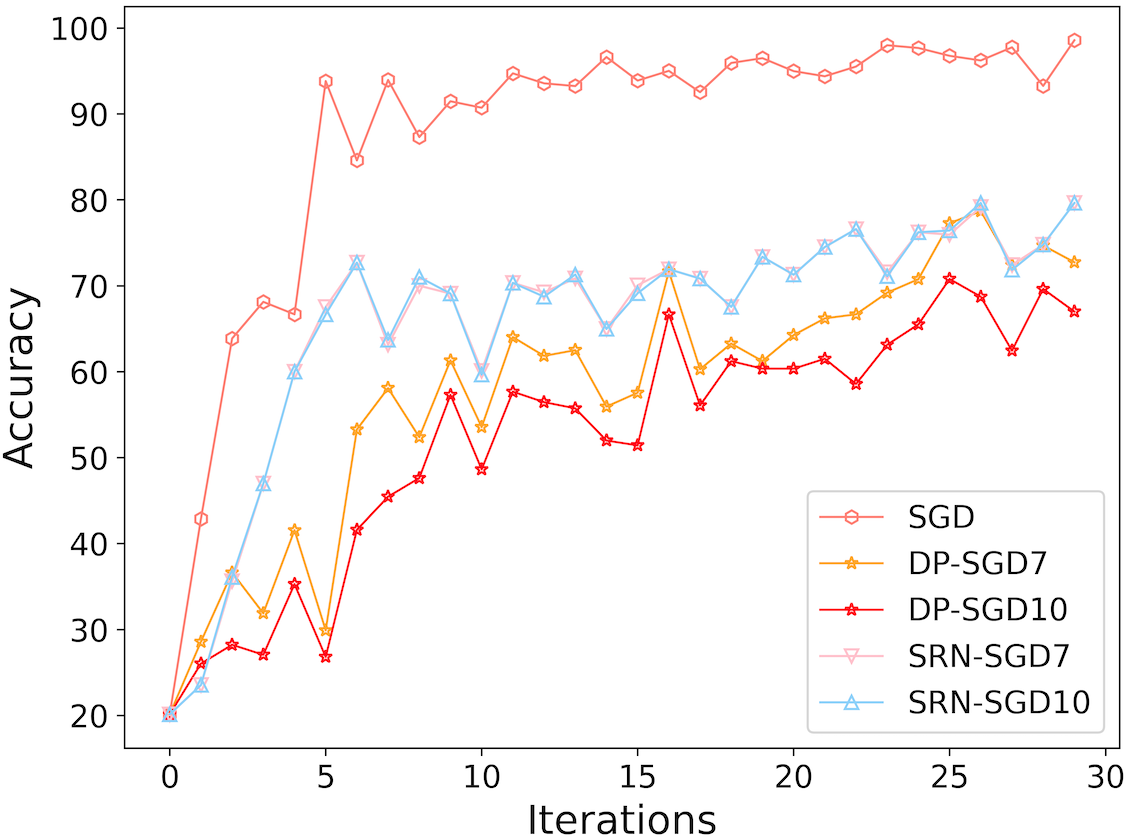}
		\subcaption{$\sigma = 7,10$}
		\label{sigma_7_10_iteration}
	\end{minipage}
 \caption{Test Accuracy on the MNIST Dataset with Different $\sigma$ in Iterations}
	\label{test_acc_mnist_sigma_iter}
 \end{figure*}
\begin{figure*}
\centering
    \begin{minipage}[t]{0.31\textwidth}
     \centering
		\includegraphics[width = 1.8in]{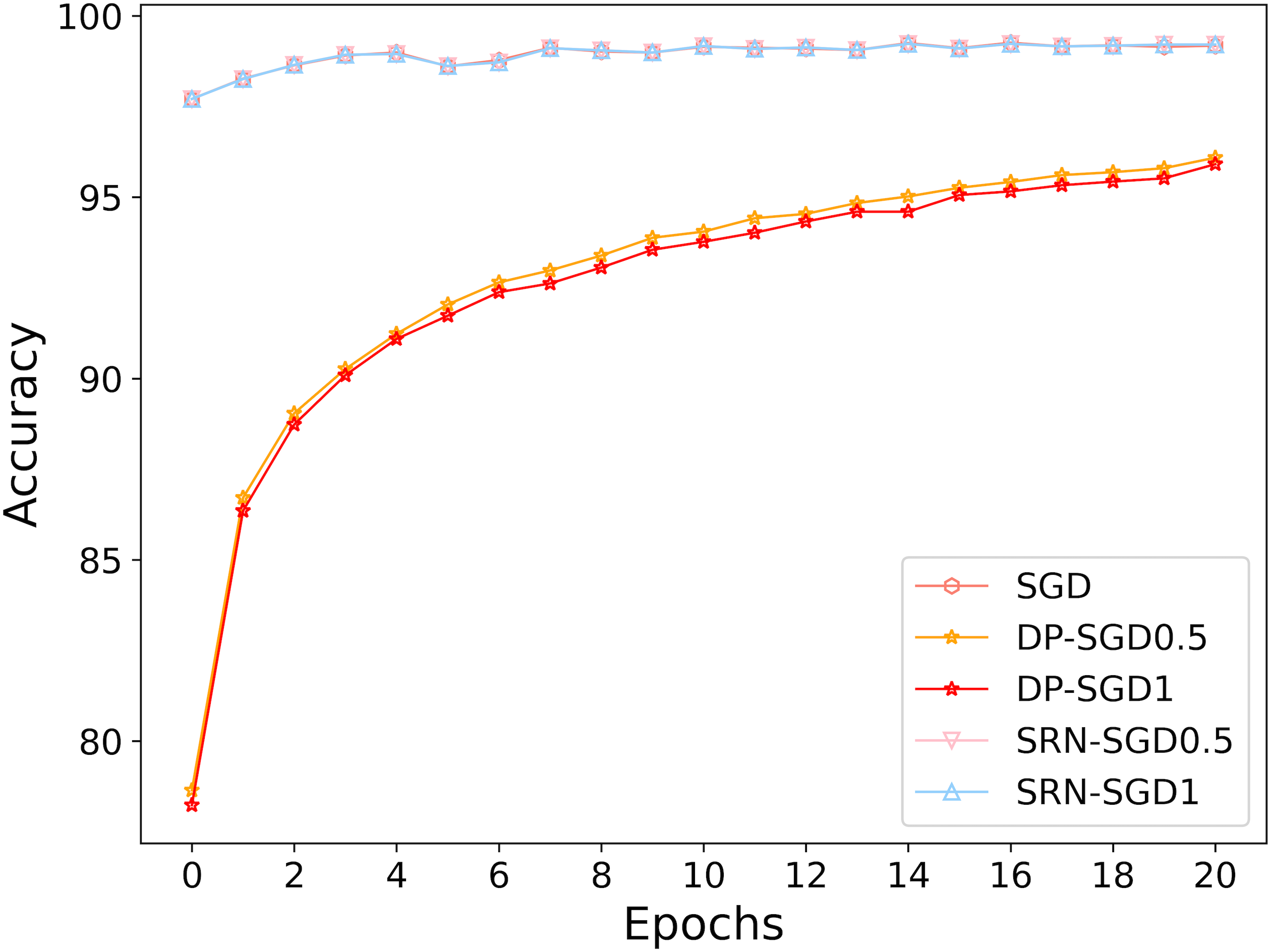}
		\subcaption{$\sigma = 0.5,1$}
		\label{sigma_0.5_1_epoch}
	\end{minipage}
	\begin{minipage}[t]{0.31\textwidth}
  \centering
		\includegraphics[width = 1.8in]{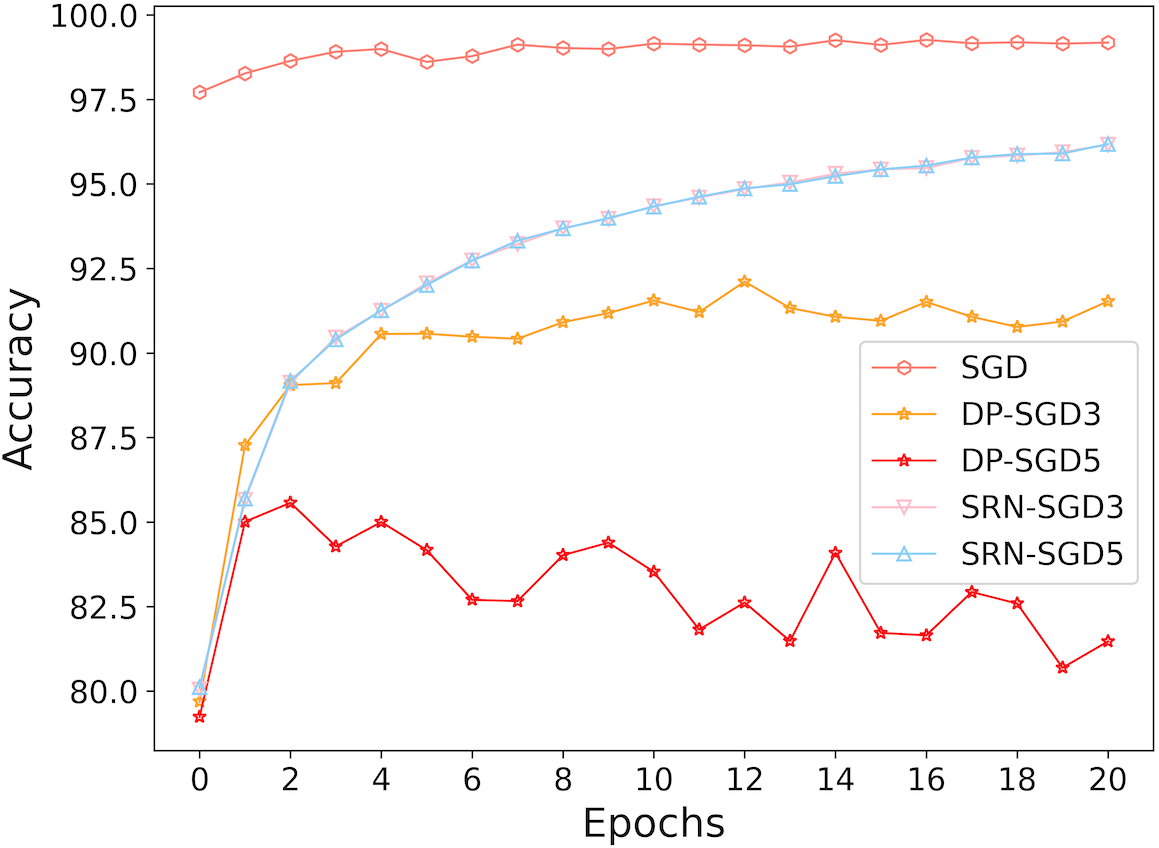}
		\subcaption{$\sigma = 3,5$}
		\label{sigma_3_5_epoch}
	\end{minipage}
	\begin{minipage}[t]{0.31\textwidth}
  \centering
		\includegraphics[width = 1.8in]{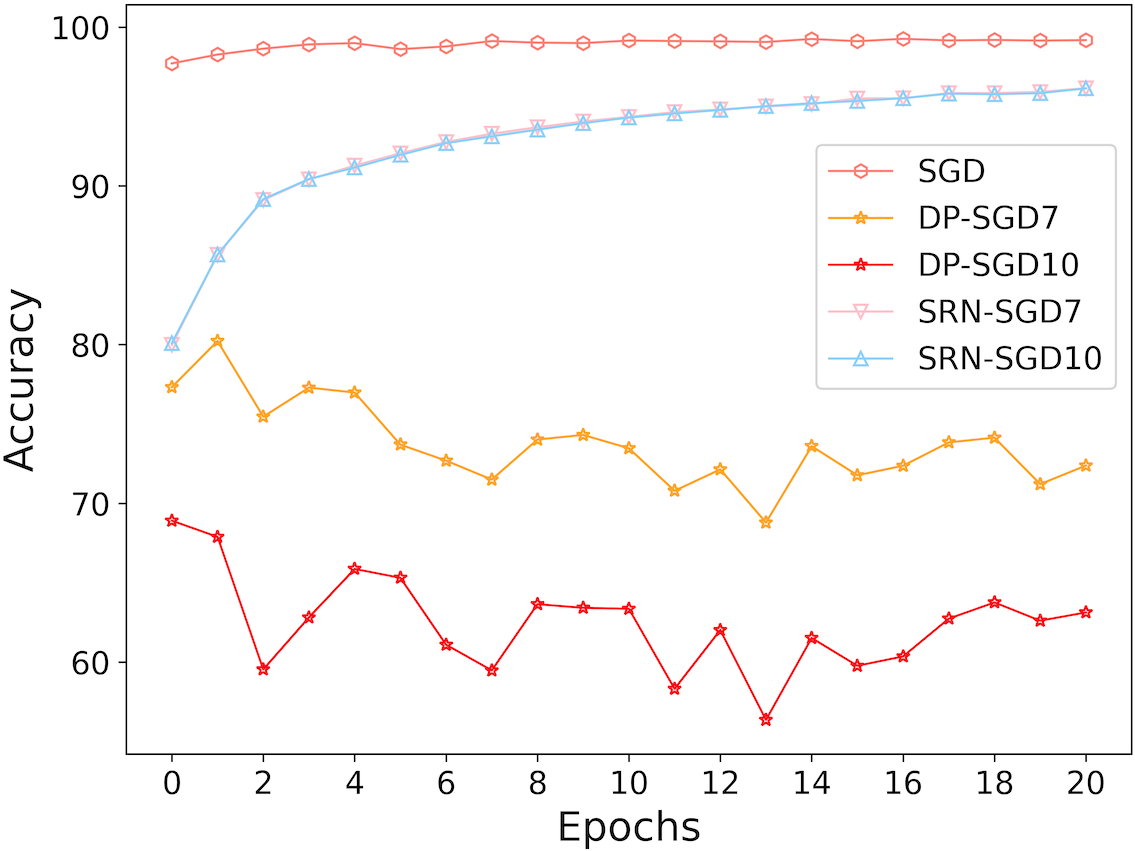}
		\subcaption{$\sigma = 7, 10$}
		\label{sigma_7_10_epoch}
	\end{minipage}
	\caption{Test Accuracy on MNIST Dataset with Different $\sigma$  in Epochs}
	\label{test_acc_mnist_sigma}
\end{figure*}


\section{Experimental Evaluation and Explanation}
\label{sec::exp}

Our experiments are conducted on a commodity PC running Ubuntu with Intel Xeon(R) E5-2630 v3 CPU, 31.3 GiB RAM, and GeForce RTX $3090$ Ti GPU.
In this section, we report the convergence/training performance and test accuracy (varying with $\epsilon$) by conducting an extensive comparison with state-of-the-arts, including  
DP-SGD~\cite{ccs/AbadiCGMMT016}, IPA-RF~\cite{nips/FeldmanZ21}, GDP~\cite{corr/abs-1911-11607}, SAS-DP~\cite{bigdataconf/ChenL20}, DP-GED~\cite{corr/abs-2007-11524}, DP-MP~\cite{sp/Yu0PGT19}, LDP~\cite{nips/GhaziNR21}, FeatureDP~\cite{iclr/TramerB21}, and FedFed~\cite{nips/YangZZ0PL023} over standard benchmark datasets.
By employing GridCam~\cite{iccv/SelvarajuCDVPB17}, we visualize differentially private training  to show the difference in representation.

\subsection{Experimental Setup}
\label{app:exp-config}
\subsubsection{Configuration and Dataset}
The baseline DP-SGD implementation is pyvacy (\url{https://github.com/ChrisWaites/pyvacy}), and , while federated learning with $\dphero$ has been realized in FedFed framework~\cite{nips/YangZZ0PL023}.
We configure experimental parameters with reference to \cite{ccs/AbadiCGMMT016}'s  setting.
To be specific, we configure lot size $L=10,50,200,400$, $\delta=1.0^{-5}\text{ or }1.0^{-6}$, and learning rate $\eta=0.1\text{ or }0.2$. 
The noise level $\sigma$ is set to be $0.5,1,3,5,7,10$ for comprehensive comparison.
Fairly, we use identical $\epsilon$ as in state-of-the-art and compare test accuracy.

Experimental evaluations are performed on the MNIST dataset~\cite{pieee/LeCunBBH98} and the CIFAR-10 dataset~\cite{krizhevsky2009learning}. 
MNIST dataset includes $10$ classes of hand-written digits of $28 \times28$ gray-scale.
It contains $60,000$ training examples and $10,000$ testing examples.
CIFAR-10 dataset contains $10$ classes  of images, of $32\times32$ color-scale with three channels, 
It contains $50,000$ in training examples and $10,000$ in testing examples.

 \subsubsection{Model Architecture} On the MNIST dataset, we use {LeNet}~\cite{pieee/LeCunBBH98}, which reaches accuracy of $99\%$ in about $10$ epochs without privacy.
On CIFAR-10, we use two {convolutional layers} followed by two fully connected layers. 
In detail, convolution layers use $5\times 5$ convolutions, followed by a ReLU and $2\times2$ max-pooling. 
The latter is flattened to a vector that gets fed into two {fully connected layers} with $384$ units. 
This architecture, non-privately, can get to about $86\%$ accuracy in  $\sim200$ epochs.

\subsection{Model Utility and Training Performance}
\label{sec::model_utility}
\subsubsection{Convergence Analysis}
Figure~\ref{test_acc_mnist_sigma_iter}, Figure~\ref{test_acc_mnist_sigma}, and Figure~\ref{acc_cifar10} show the process of convergence on the MNIST and CIFAR-10 datasets in iterations and epochs when $\sigma=0.5,1,3,5,7,10$, respectively. 
The epoch-based figures show the whole training process on two datasets, while the iteration-based figures only display the first $30$ iterations meticulously due to $x$-axis length limitation.

For the very-tiny noise level $\sigma=0.5,1$, $\dphero$ SGD reaches an almost identical convergence route as pure SGD when training over the MNIST dataset. 
For DP-SGD, iteration-wise accuracy decreases at the start of training.
For a relatively small noise level $\sigma=3,5$, we can see that $\dphero$ SGD  converges more stable.
Although $\dphero$ SGD can not reach the identical accuracy as pure SGD, its shape (e.g., from iteration=$[5,10]$ and epoch=$[4,20]$) of convergence is much more similar to SGD than DP-SGD.
For $\sigma\geq5$, the convergence of DP-SGD turns out to be very unstable, while $\dphero$ SGD looks more robust.
Besides, the shaking of $\dphero$ SGD is also relatively smaller, which contributes to step-wise stability during a whole training process. 

On CIFAR-10, Figure~\ref{acc_cifar10} shows the test accuracy by training from scratch.
Recall that DP-SGD over CIFAR-10 typically requires a pretraining phase.
For $\sigma=0.5,1$,  $\dphero$ SGD attains competitive training convergence compared with SGD training. 
For $\sigma=3,5$, $\dphero$ SGD training still moves towards convergence, while DP-SGD could not.
For $\sigma=7,10$, both $\dphero$ SGD and DP-SGD could not converge, whereas $\dphero$ SGD collapses later.

\begin{figure*}[h]
	\centering
 \begin{minipage}[t]{0.31\textwidth}
		\centering
		\includegraphics[width = 1.8in]{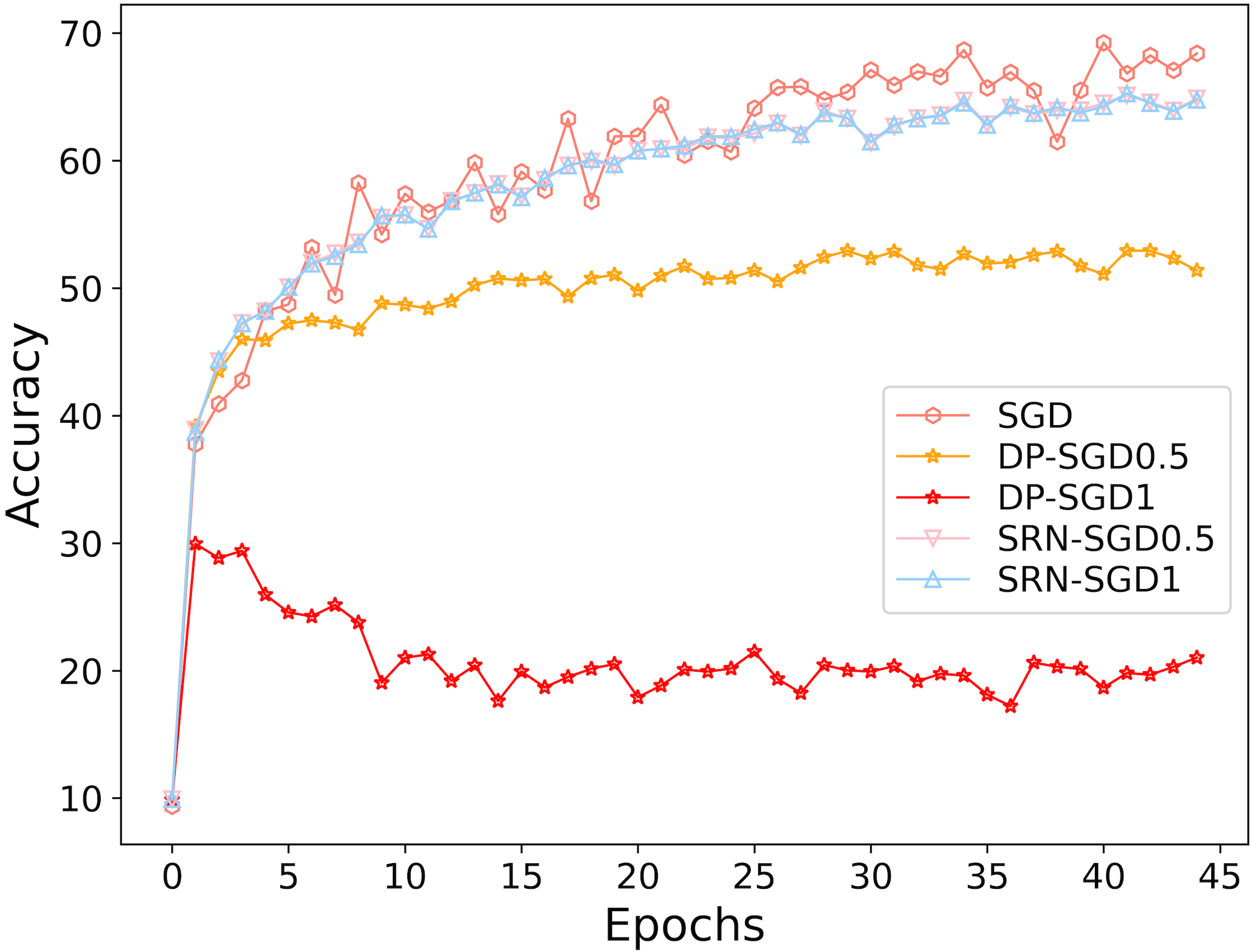}
		\subcaption{$\sigma=0.5,1$}
		\label{cifarsigma0.51}
	\end{minipage}
	\begin{minipage}[t]{0.31\textwidth}
		\centering
		\includegraphics[width = 1.8in]{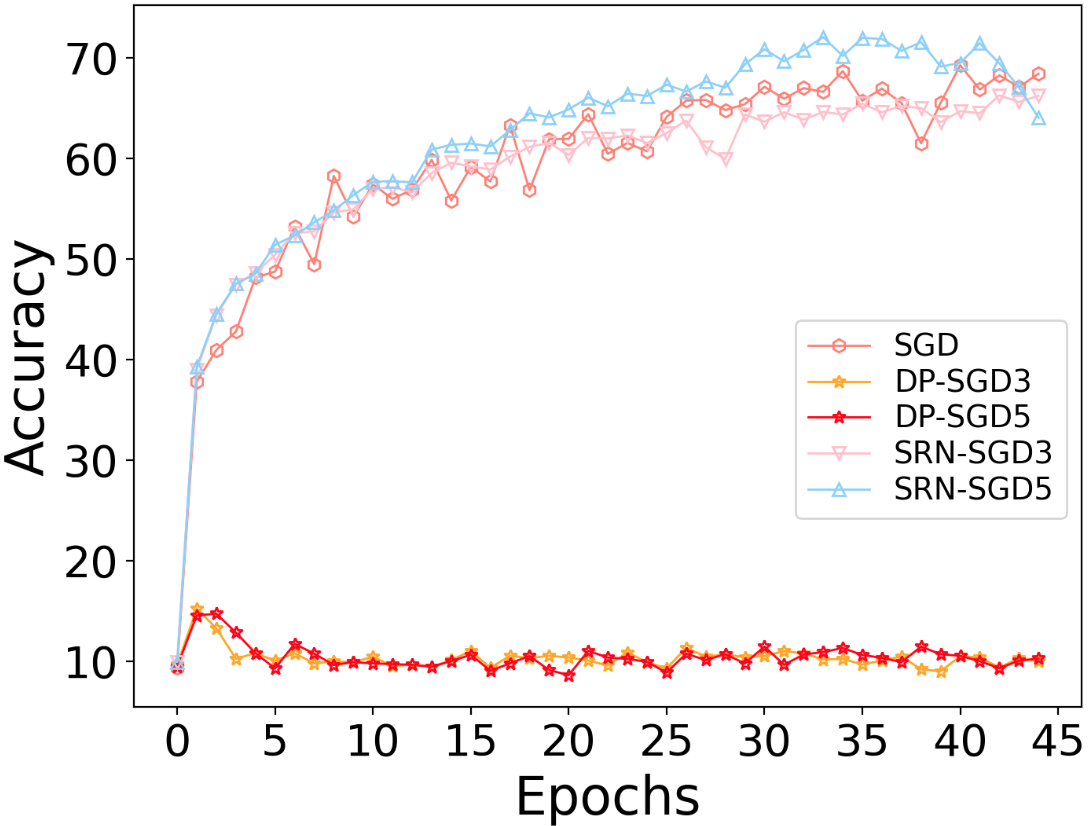}
		\subcaption{$\sigma=3,5$}
		\label{cifarsigma35}
	\end{minipage}
	\begin{minipage}[t]{0.31\textwidth}
	\centering
		\includegraphics[width = 1.8in]{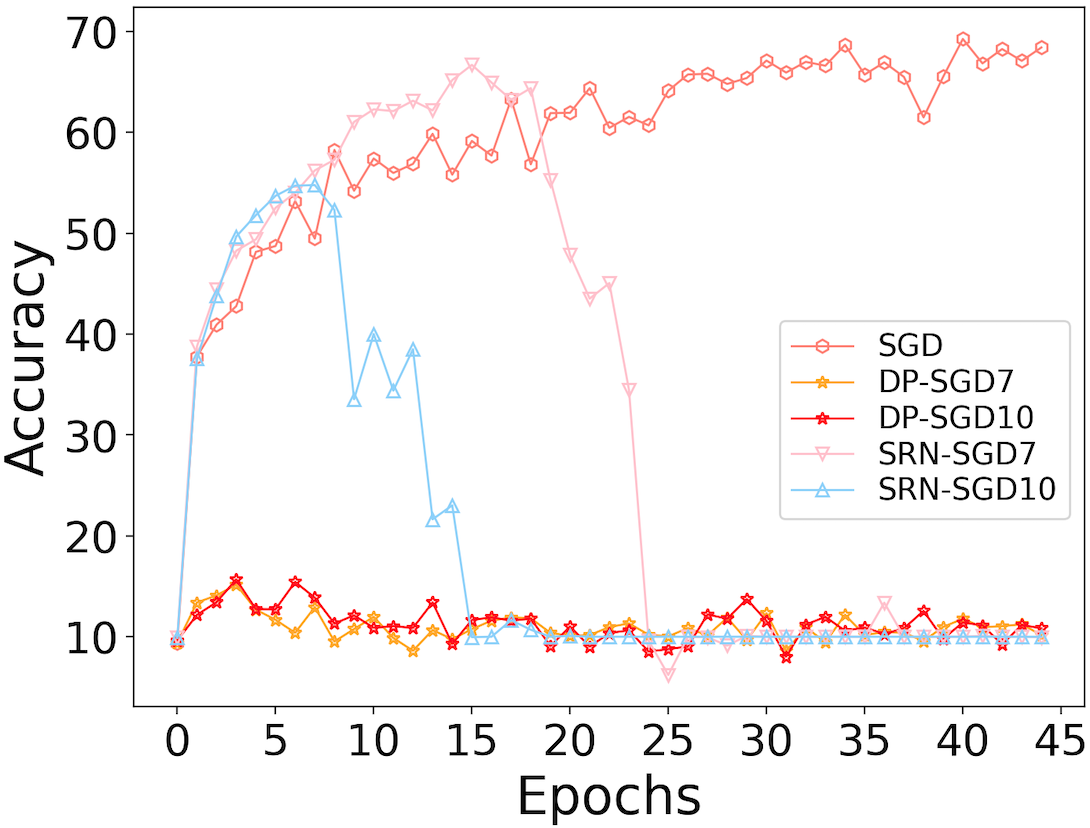}
		\subcaption{$\sigma=7,10$}
		\label{cifarsigma710}
	\end{minipage}
	\caption{Test Accuracy on CIFAR-10 with Different $\sigma$  in Epochs}
	\label{acc_cifar10}
\end{figure*}

\subsubsection{Model Accuracy}
Table~\ref{tab::testacc} shows comparative results with prior works.
To be fair, we compare the test accuracy of the trained models under the constraint of identical $\epsilon$.
We can see that $\dphero$ improves the test accuracy of 
state-of-the-arts~\cite{ccs/AbadiCGMMT016,nips/FeldmanZ21,corr/abs-1911-11607,bigdataconf/ChenL20,corr/abs-2007-11524,sp/Yu0PGT19,nips/GhaziNR21,iclr/TramerB21}.
In most cases, our $\dphero$ SGD could attain $>98\%$ test accuracy on the MNIST dataset, whereas other works achieve $95\%\sim97\%$.
Only several works were trained over the CIFAR-10 dataset, yet with the $<60\%$ accuracy. 
In contrast, $\dphero$ SGD could achieve $>64.5\%$ accuracy, showing much better results.

Specifically, $\dphero$ SGD improves $18\%$ accuracy on \cite{corr/abs-2007-11524}, $47\%$ accuracy on \cite{sp/Yu0PGT19}, and $22\%$ accuracy on \cite{bigdataconf/ChenL20}.
Training a DP model over the CIFAR-10 dataset may require a pretraining phase, whereas $\dphero$ SGD training could alleviate this.
It shows that $\dphero$ SGD behaves better on more representative datasets (e.g., CIFAR-10$>$MNIST) than DP-SGD.
Figure~\ref{acc_epsi} shows a box-whisker plot on accuracy given varying $\epsilon$.
Except for following identical configuration of $\epsilon$, we show additional results with $\epsilon=1,2,3,4$.
The test accuracy is relatively stable for different $\epsilon$ in different training processes.
When $\epsilon$ is very large, although test accuracy is high, DP protection may not be sufficient for practical usage.
Experimental results show that $\dphero$ SGD is more robust against large noise and supports faster convergence, especially for representative datasets.

When extending to FedFed, accuracy is sensitive to the degree of data heterogeneity (non-IID-ness) among clients. 
Under the IID scenario, where each client’s local data distribution mirrors the global distribution, the global model can closely match the centralized $98\%$ approximately over MNIST, typically reaching $97\%$-$98\%$ accuracy. However, as data becomes more heterogeneous, such as when clients have unbalanced or partially disjoint label distributions, overall accuracy declines. Moderate non-IID partitions (Dirichlet distribution parameter $\alpha>10$) generally yield $95\%$-$97\%$ accuracy, while higher heterogeneity ($\alpha<10$) can reduce accuracy to $90\%$-$95\%$. 

	\begin{table}[!t]
		\caption{Test Accuracy Compared with Prior Top-tier Works}
		\label{tab::testacc}
			\setlength\tabcolsep{3pt}
			\centering
				\begin{tabular}{c|c|c|c|c}
				\hline
				&  && &  $\dphero$ SGD  \\  
				\multirow{-2}{*}{\textbf{Dataset}} &   \multirow{-2}{*}{\textbf{Work}}      & \multirow{-2}{*}{\textbf{$\epsilon$}}   & \multirow{-2}{*}{\textbf{Accuracy}}      & \textbf{Accuracy    }       \\ \hline
				& & $2$ & $95\%$ & $\mathbf{98.10\%}$ \\
				& & $8$ & $97\%$ &  $\mathbf{98.11\%}$ \\
				& \multirow{-3}{*}{DP-SGD} & $\infty$ &${98.3\%}$ &  $98.12\%$ \\
				\cline{2-5}
				&                  &     $1.2$      &   $96.6$        &   $\mathbf{98.13\%}$             \\
				&\multirow{-2}{*}{IPA-RF} & $3$ & $97.7\%$ & $\mathbf{98.11\%}$\\
				\cline{2-5}
				&             &      $2.32$    &   $96.6\%$  &    $96.18\%$                    \\
				&      \multirow{-2}{*}{GDP}       &      $5.07$    &   $97.0\%$  &  $\mathbf{98.10\%}$                 \\
				\cline{2-5}
				&      SAS-DP         &    $2.5$     &    $90.0\%$       &    $\mathbf{98.12\%}$            \\
				&      DP-GED          &    $3.2$       &   $96.1\%$       &    $\mathbf{98.11\%}$                \\
				&      DP-MP        &      $6.78$    &  $93.2\%$         &    $\mathbf{98.10\%}$              \\
				\cline{2-5}
				& &$1$& $95.82\%$ & $\mathbf{98.11\%}$\\
				& \multirow{-2}{*}{LDP}&$2$& $98.78\%$ &$98.10\%$\\
				\cline{2-5}
				&                                                                  &         $1.2$, $2.0$                     &               &           $\mathbf{98.13\%}$             \\
				\multirow{-13}{*}{{MNIST}}  &      \multirow{-2}{*}{FeatureDP}       &   $2.5,2.9$   &  \multirow{-2}{*}{$\approx 98\%$}         &    $\mathbf{98.12\%}$         \\\hline
				&   DP-GED        &    $3.0$       &   $55.0\%$       &     $\mathbf{64.93\%}$          \\
				&    DP-MP        &      $6.78$    &  $44.3\%$         &      $\mathbf{65.04\%}$                 \\
				\multirow{-3}{*}{{CIFAR-10}}       &      SAS-DP         &    $8.0$     &    $53.0\%$       &       $\mathbf{65.12\%}$             \\\hline
				
			\end{tabular}
		
	\end{table}
\begin{figure*}[h]
	\centering
 \begin{minipage}[b]{0.31\textwidth}
		\centering
		\includegraphics[width = 1.8in]{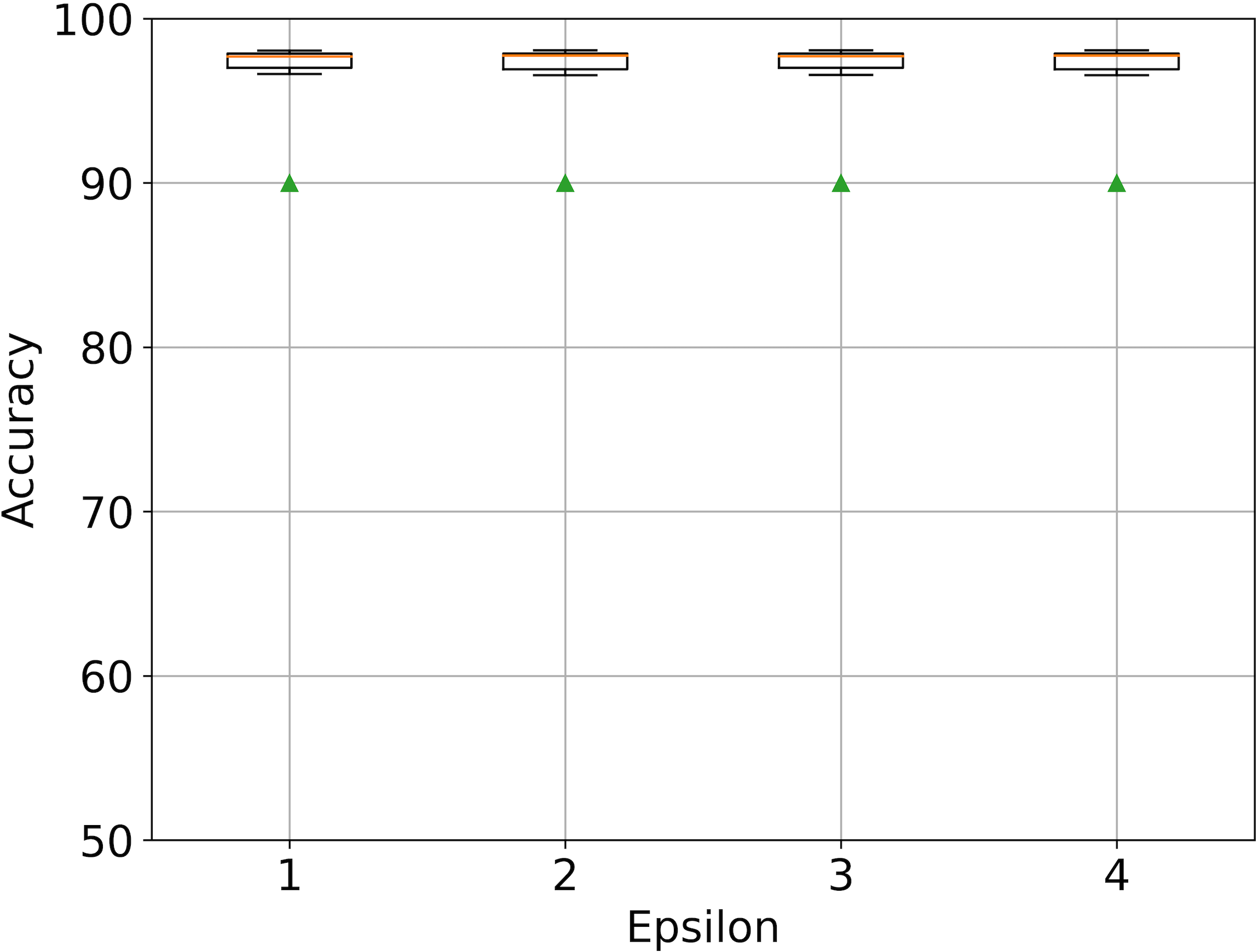}
		\subcaption{Test Accuracy with Preset $\epsilon$}
		\label{acc_1epsilon_com}
	\end{minipage}
	\begin{minipage}[b]{0.31\textwidth}
		\centering
		\includegraphics[width = 1.8in]{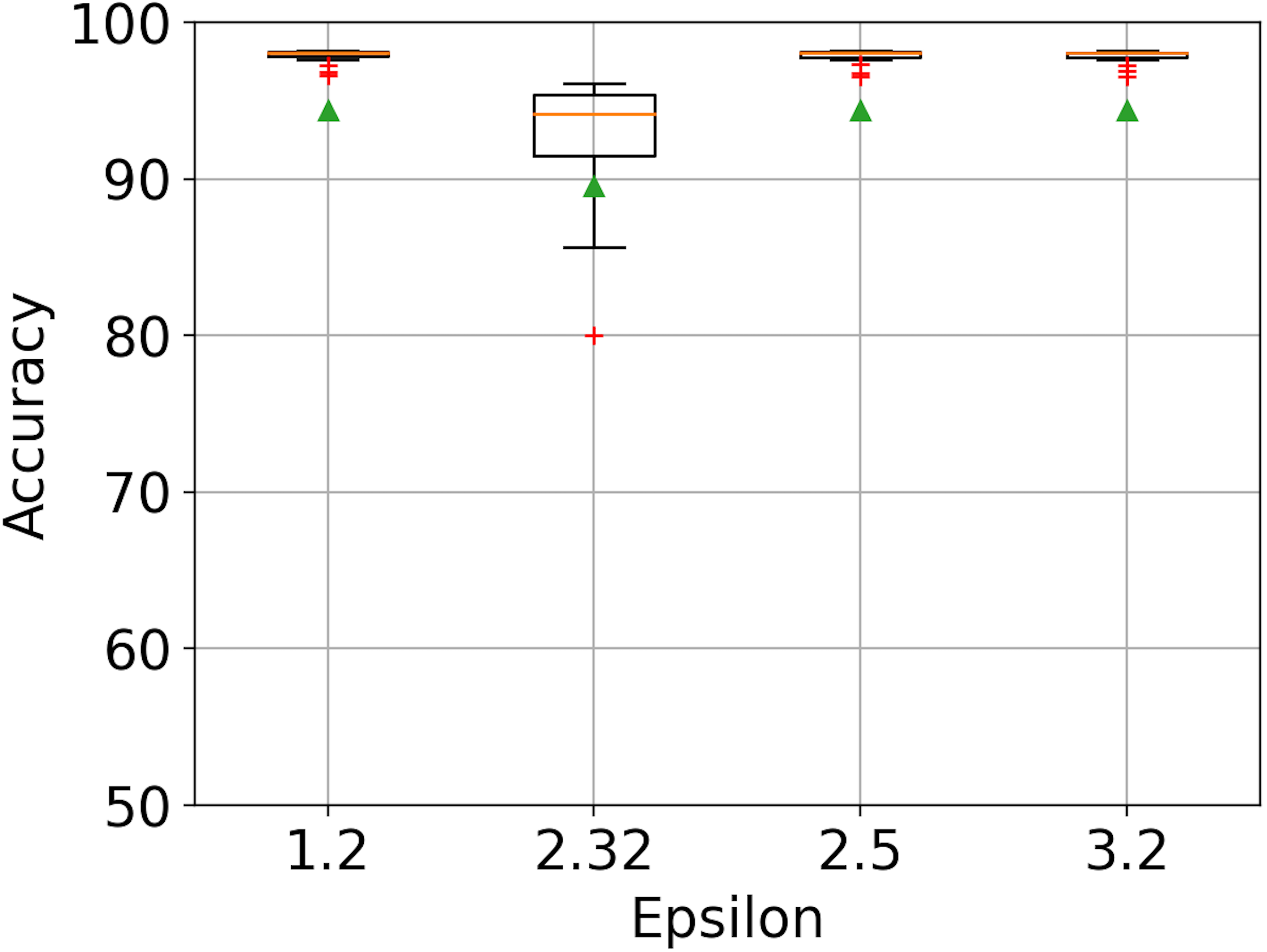}
		\subcaption{Test Accuracy with Small $\epsilon$}
		\label{acc_1epsilon}
	\end{minipage}
	\begin{minipage}[b]{0.31\textwidth}
		\centering
		\includegraphics[width = 1.8in]{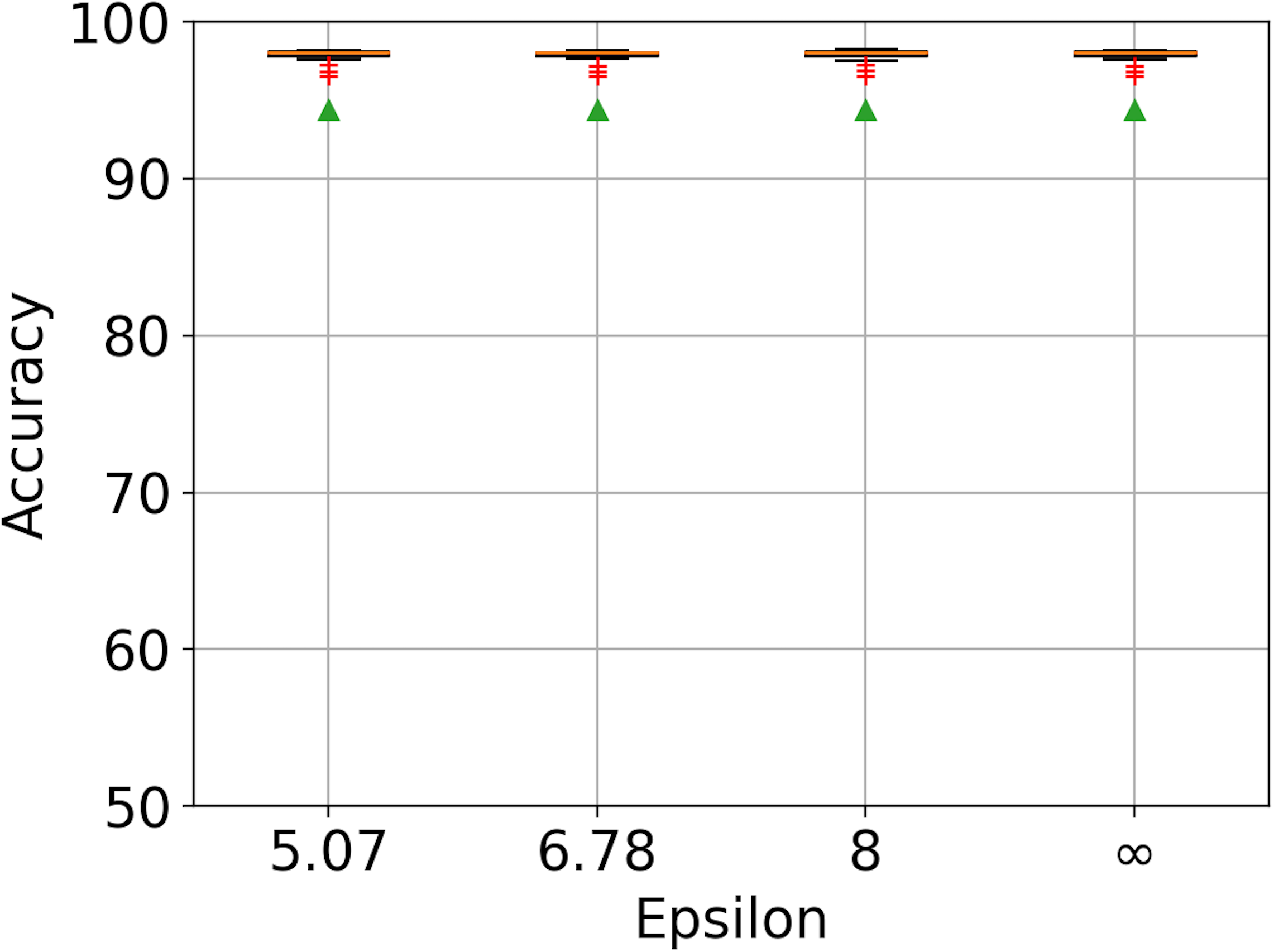}
		\subcaption{Test Accuracy with Large $\epsilon$}
		\label{acc_2epsilon}
	\end{minipage}
	\caption{Test Accuracy on MNIST Dataset With Varying $\epsilon$}
	\label{acc_epsi}
\end{figure*}

\begin{figure*}[h]
\centering
     \begin{minipage}[t]{0.31\textwidth}
     \centering
		\includegraphics[width = 1.8in]{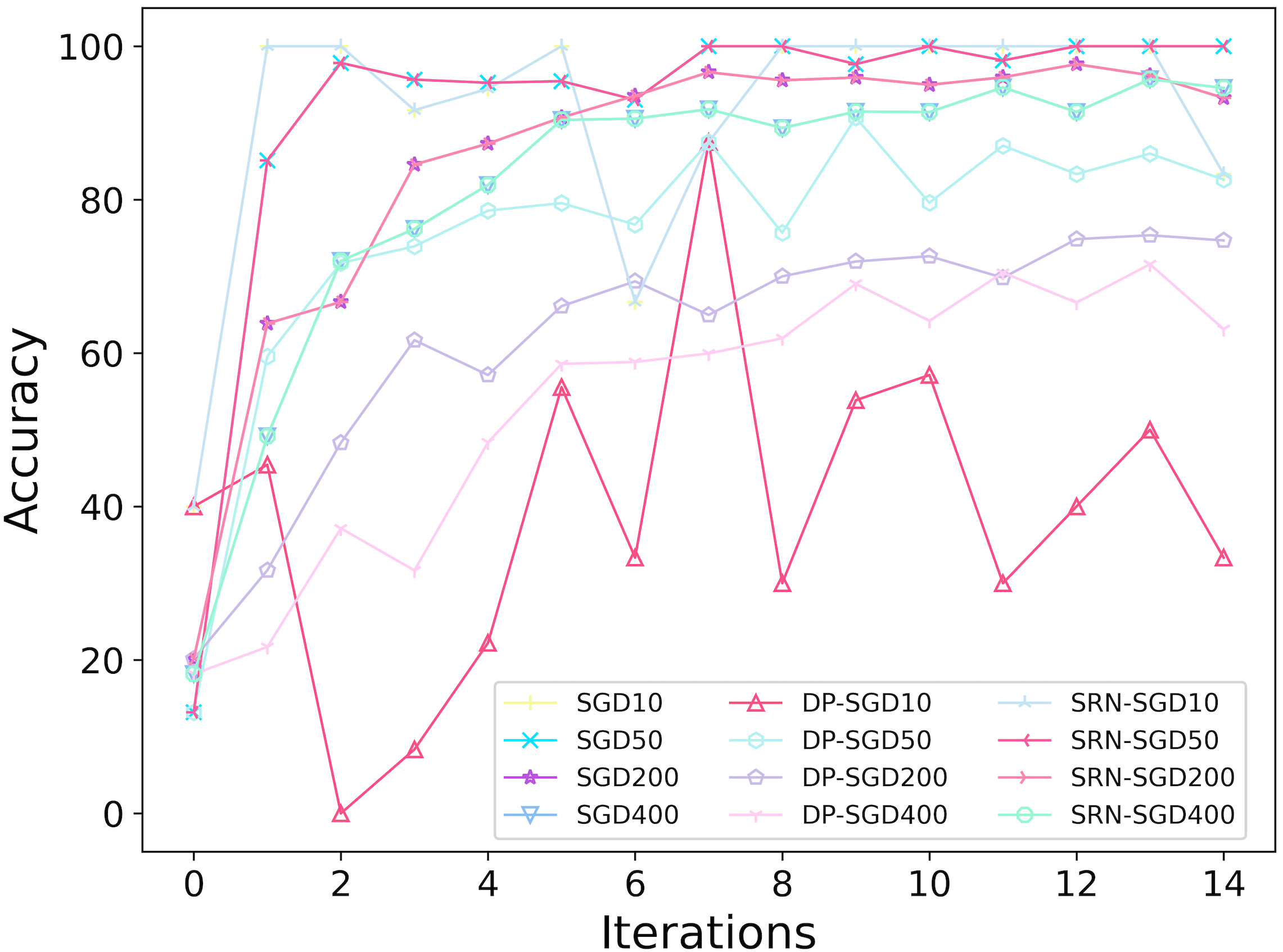}
		\subcaption{$L=10,50,200,400,\sigma=1$}
		\label{mnist-lot-iterations-sigma1}
	\end{minipage}
	\begin{minipage}[t]{0.31\textwidth}
 \centering
		\includegraphics[width = 1.8in]{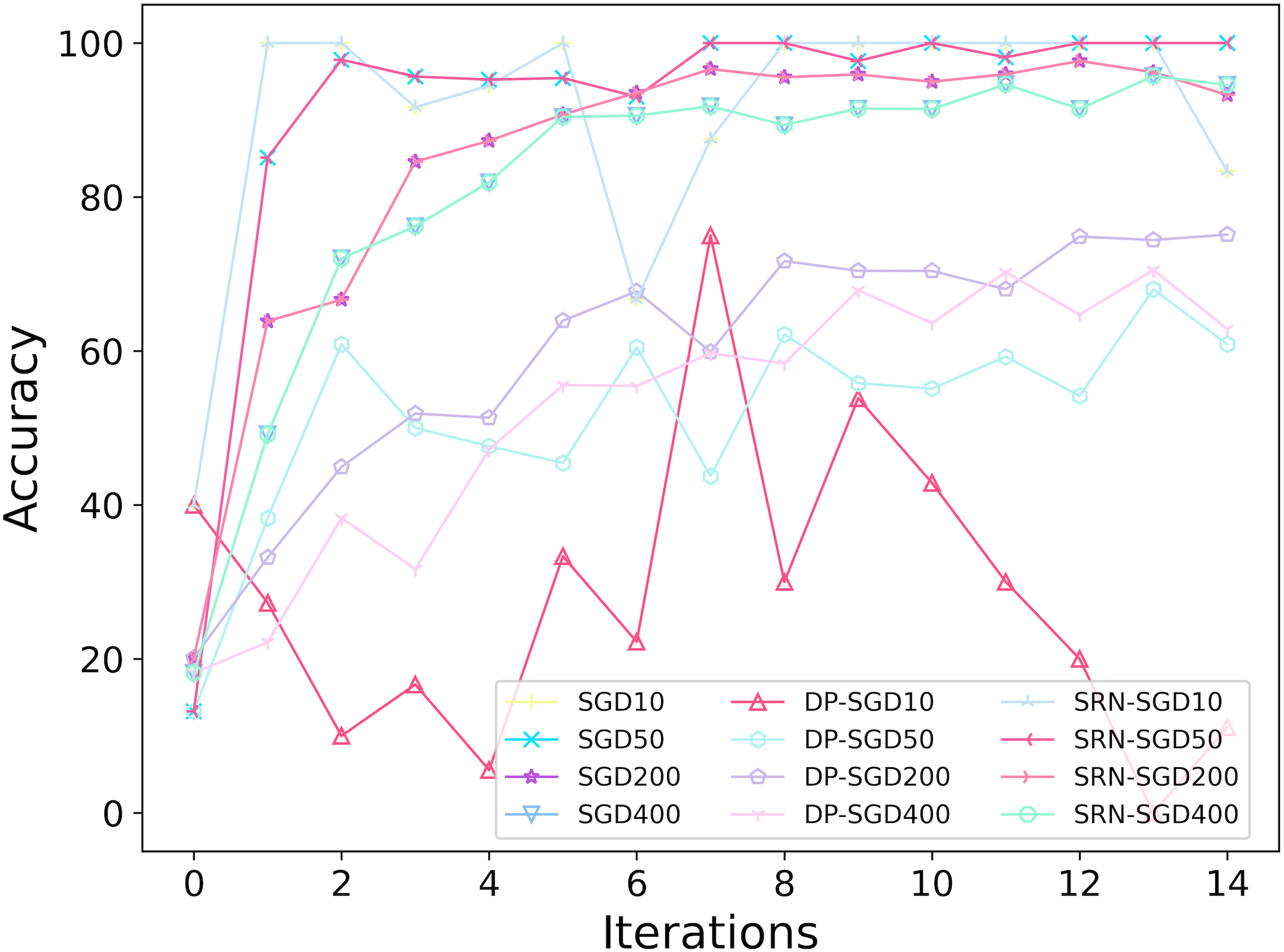}
		\subcaption{$L=10,50,200,400,\sigma=3$}
		\label{mnist-lot-iterations-sigma3}
	\end{minipage}
	\begin{minipage}[t]{0.31\textwidth}
 \centering
		\includegraphics[width = 1.8in]{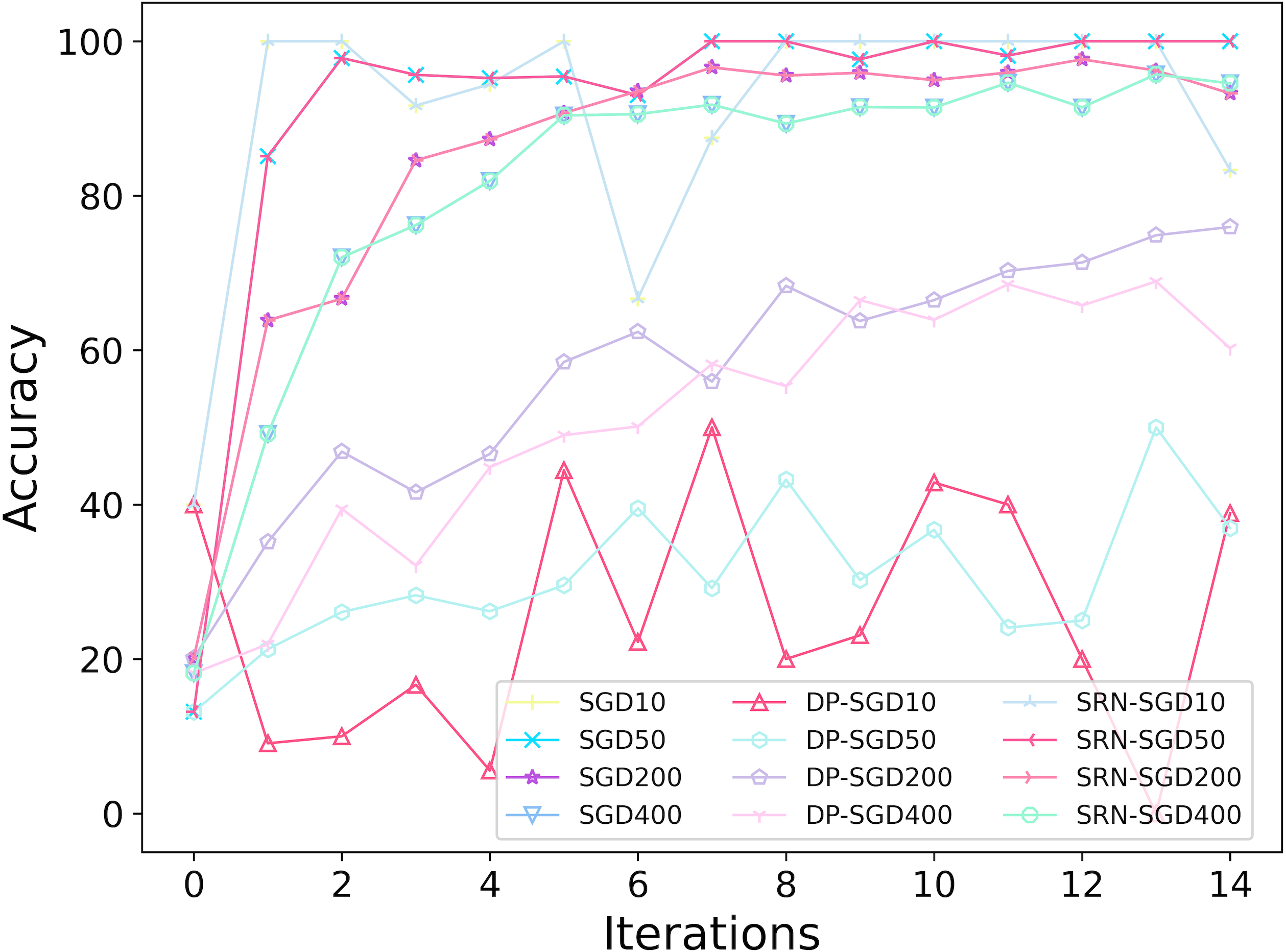}
		\subcaption{$L=10,50,200,400,\sigma=5$}
		\label{mnist-lot-iterations-sigma5}
	\end{minipage}
	\caption{Test Accuracy on MNIST Dataset with Varying Lot Sizes}
	\label{mnist-lot-iterations-sigma}
\end{figure*}

\begin{figure*}[h]
	\centering
   \begin{minipage}[t]{0.3\textwidth}
		    \includegraphics[width = 1.8in]{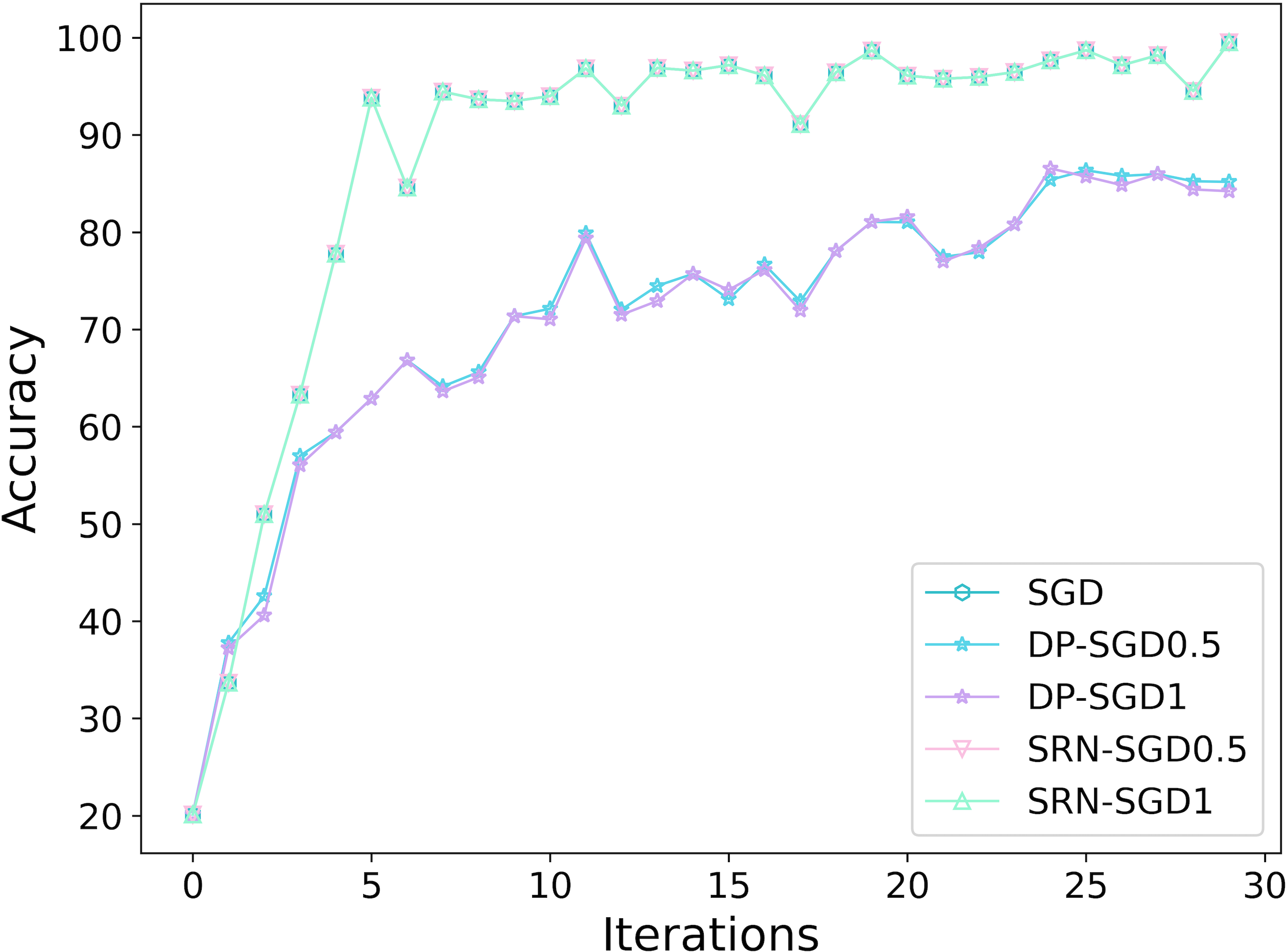}
		\subcaption{$\sigma=0.5,1$}
		\label{mnist-lr0.2-sigma-iterations-0.5-1}
	\end{minipage}
	\begin{minipage}[t]{0.3\textwidth}
		\centering
		\includegraphics[width = 1.8in]{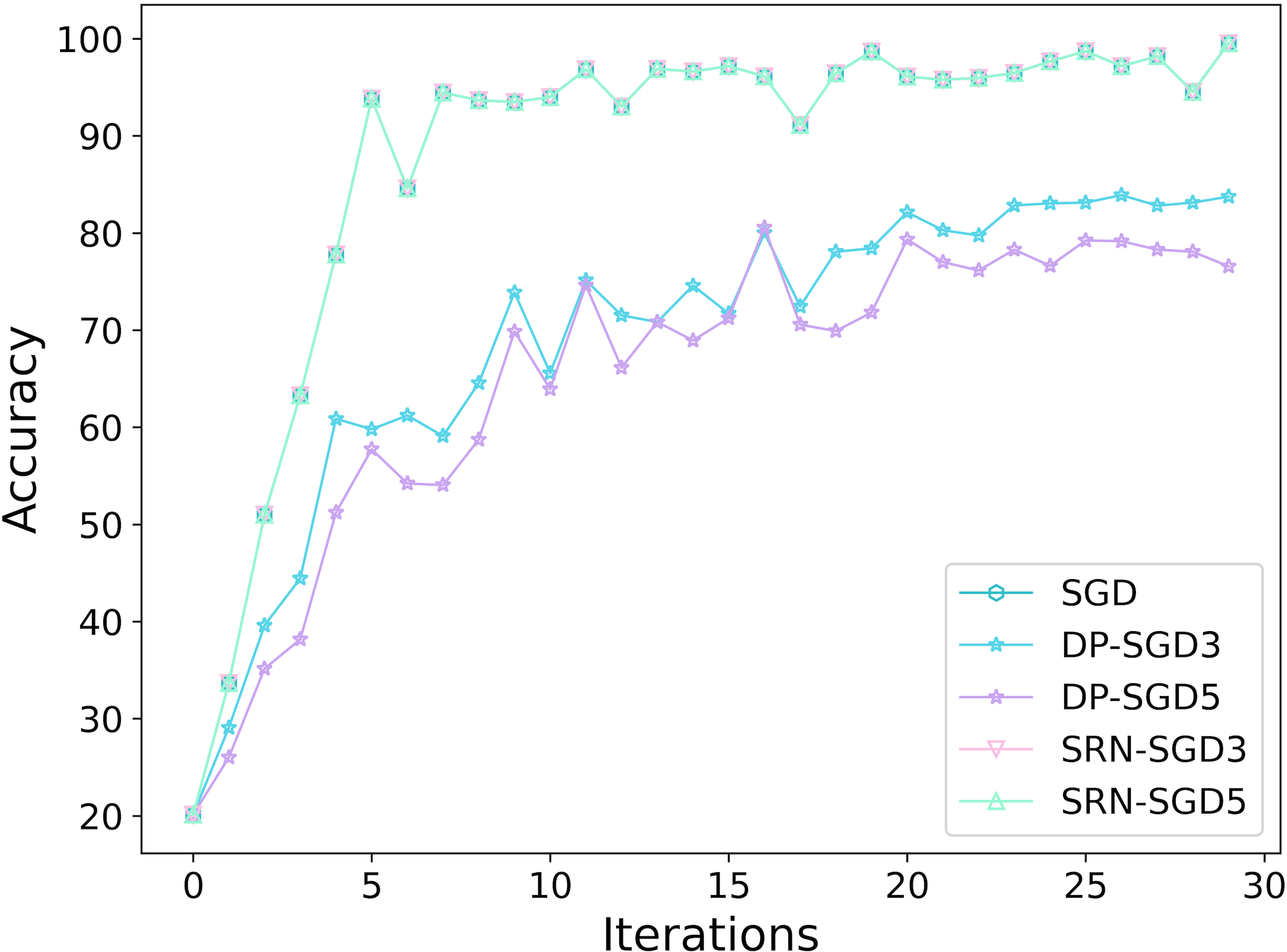}
		\subcaption{$\sigma=3,5$}
		\label{mnist-lr0.2-sigma-iterations-3-5}
	\end{minipage}
	\begin{minipage}[t]{0.3\textwidth}
	\centering
		\includegraphics[width = 1.8in]{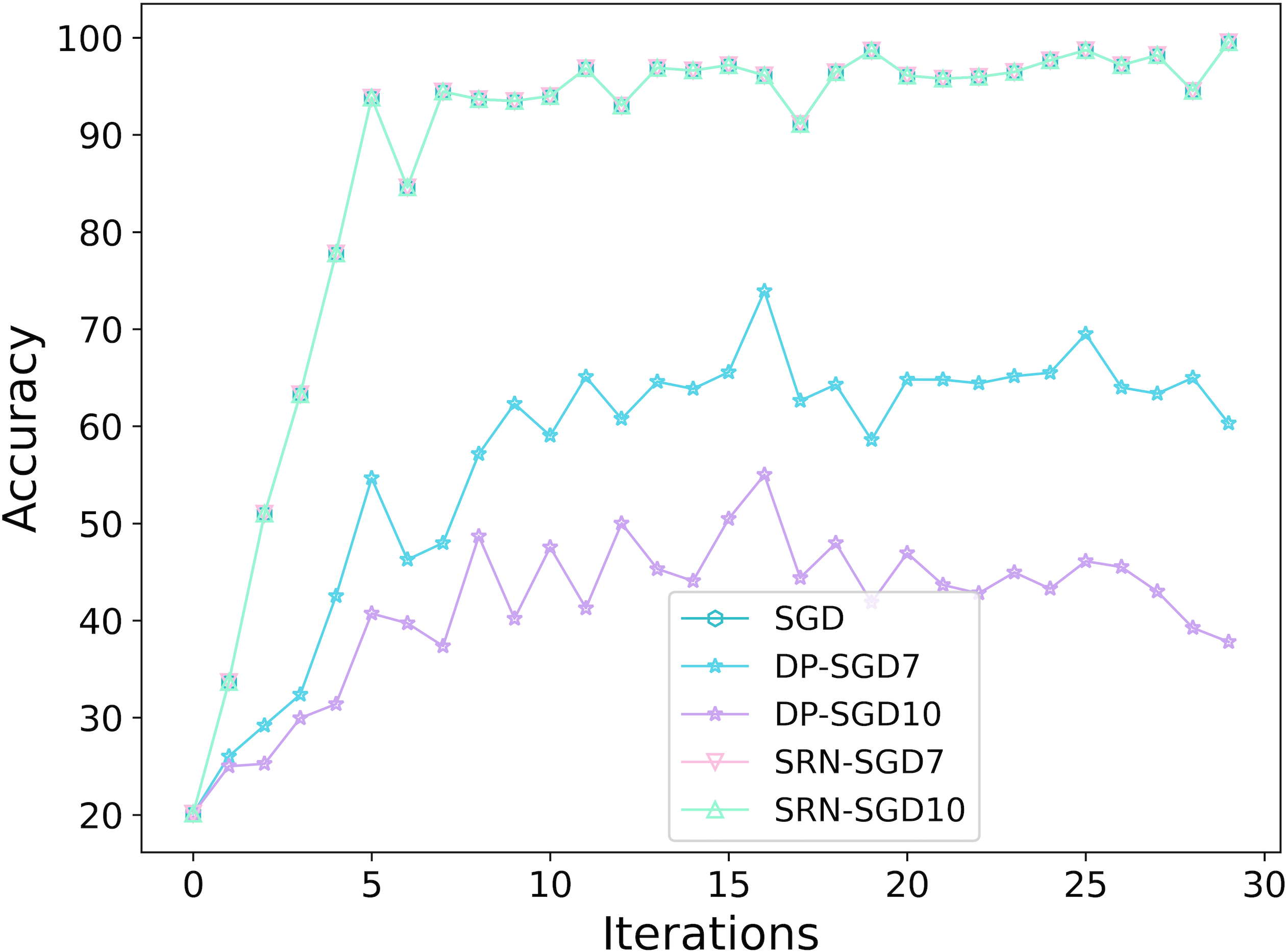}
		\subcaption{$\sigma=7,10$}
		\label{mnist-lr0.2-sigma-iterations-7-10}
	\end{minipage}
	\caption{Test Accuracy on MNIST Dataset with Learning Rate $\eta=0.2$}
	\label{mnist-lr0.2-sigma-iterations}
\end{figure*}

\subsection{Explaining Experiments}
Explainable AI (XAI) has been proposed to explain why they predict what they predict.
We adopt XAI to interpret the superiority/failure of various models by decomposing them into intuitive components by tracking and understanding the training performance, and visualizing the features.

\subsubsection{Tracking Initial-Phase Training}
To explain why $\dphero$ SGD converges better, we plotted the training convergence process in the initial phase, in which the trained model is near the random initialization. 
Figure~\ref{mnist-lot-iterations-sigma} displays training convergence with varying lot sizes, while Figure~\ref{mnist-lr0.2-sigma-iterations} shows training convergence when the learning rate increases to $0.2$.
Both Figure~\ref{mnist-lot-iterations-sigma} and Figure~\ref{mnist-lr0.2-sigma-iterations} confirm that $\dphero$ SGD tracks the SGD training tracks more tightly in the very beginning. 
Recall that a typical model training starts from the random initialization towards a stable status, which means fewer features are learned in the beginning.
Thus, we expect relatively less noise to protect the ``randomized'' model, which learns a limited number of features, to mitigate and destroy the typical training convergence.
Combining with Figure~\ref{cifarsigma710}, we know that model collapse would happen when sufficient noise is assigned to enough features learned from the training data.

\subsubsection{Visualizing DP Training}
Given high-resolution and precise class discrimination, we apply Grad-CAM~\cite{iccv/SelvarajuCDVPB17} to show visual results on DP training.
In Grad-CAM~\cite{iccv/SelvarajuCDVPB17}, the class-discriminative localization map of width $u$ and height $v$ for any class $c$ is defined to be $\mathcal{L}_{\mathsf{GradCAM}} = \mathsf{ReLU}(\sum_k \alpha_k^cA^k)$.
Here, the weight $\alpha_k^c$ represents a partial linearization of the  downstream feature map activation $A$.
In our experiments, we adopt GridCam~\cite{iccv/SelvarajuCDVPB17} for interpreting/visualizing how DP noise affects model training.
In a model training process, GridCam is employed to visualize explanations of learning features, with or without DP noise.

GridCam~\cite{iccv/SelvarajuCDVPB17} can coarsely locate the important regions in the image for predicting the concept, e.g.,  ``dog'' in a classification network.
Figure~\ref{heat_dog} visualizes the heat map of  training with $\dphero$ SGD compared with Figure~\ref{heat_dog_nonoise}.
$\dphero$ SGD training still maintains the representation ability to locate the important objects.
That is, the reason for more satisfying accuracy is that  the noise added to the gradients could not affect on models' ability for relatively accurate visualization in a statistical manner, \ie, still protecting individual privacy.

\subsubsection{A Practical View of Privacy Parameters}
Theoretically, DP-SGD allows setting different clipping thresholds $C$ and noise scales $\sigma$ with varying numbers of training iterations $t$ or different layers.
Although its experiments adopt the fixed value $\sigma^2= {2/\epsilon^2}\cdot\log ({1.25}/{\delta})$, $\dphero$ SGD puts a step forward, showing a practical view of adjusting $\sigma$ in every iteration and diverse noise allocation regarding every gradient update.
The added noise is typically sampled from a noise distribution parameterized by $\sigma$. 
Besides, to explore the varying $\sigma$ over diverse features, $\dphero$ SGD still adopts a constant clipping value $\mathsf{Cp}$ as in DP-SGD.

$\dphero$ SGD  assigns $\sigma$ as a variable during DP training. 
As for unbiased noise distribution, $\mu=0$ holds at every execution of sampling noise.
In probability theory, the sum of multiple independent normally distributed random variables is normal, and its variance is the sum of the two variances.
We use this conclusion to assign the $\sigma$ over diverse features at each training iteration $t$.
If we regard all assigned $\sigma$ at each iteration as a matrix, all entries in this matrix vary at different iterations.
The parameter configuration at every iteration follows Theorem~\ref{the:dpsgd-para}, supporting linearity relation to value in $\svec$ in $\dphero$ SGD.
 Although the theoretical expectation of introducing Gaussian noise with $0$ mean value remains identical to the clean model, practical training shifts the expected results to some extent. 

 \begin{figure}[!h]
	\centering
	\includegraphics[width = 0.43\textwidth]{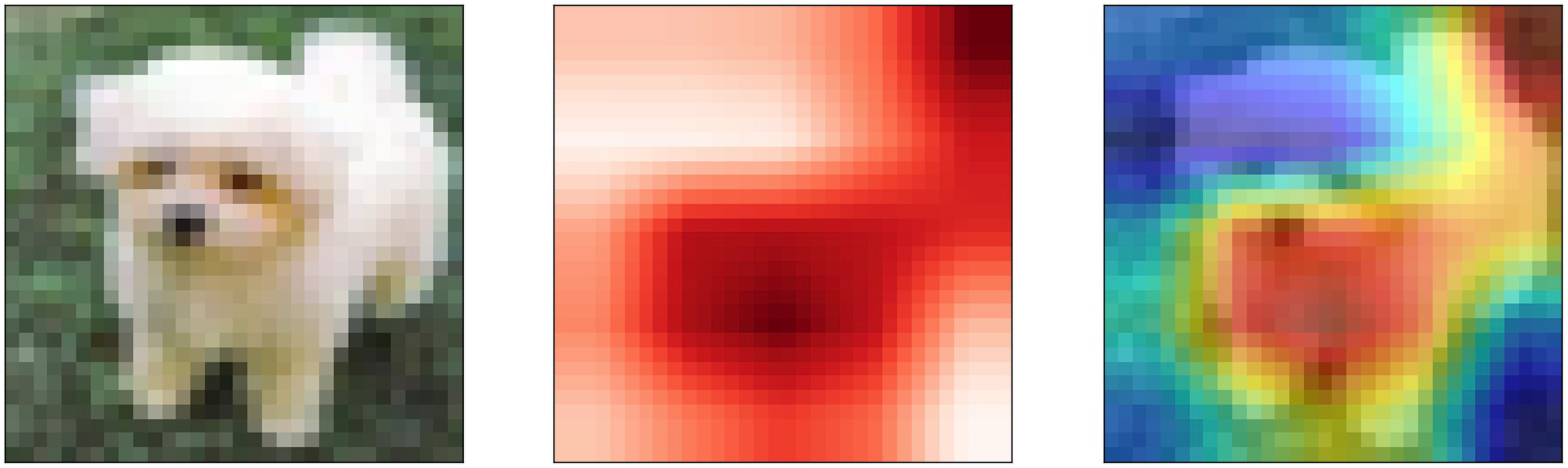}\hfil
	\includegraphics[width = 0.43\textwidth]{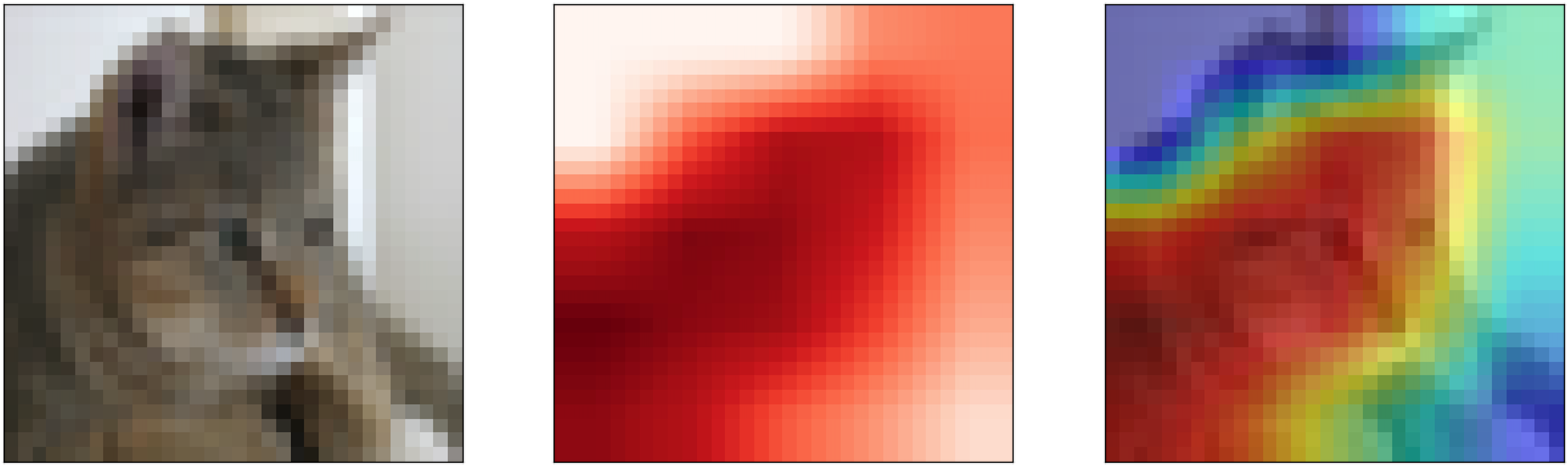}
	\includegraphics[width = 0.43\textwidth]{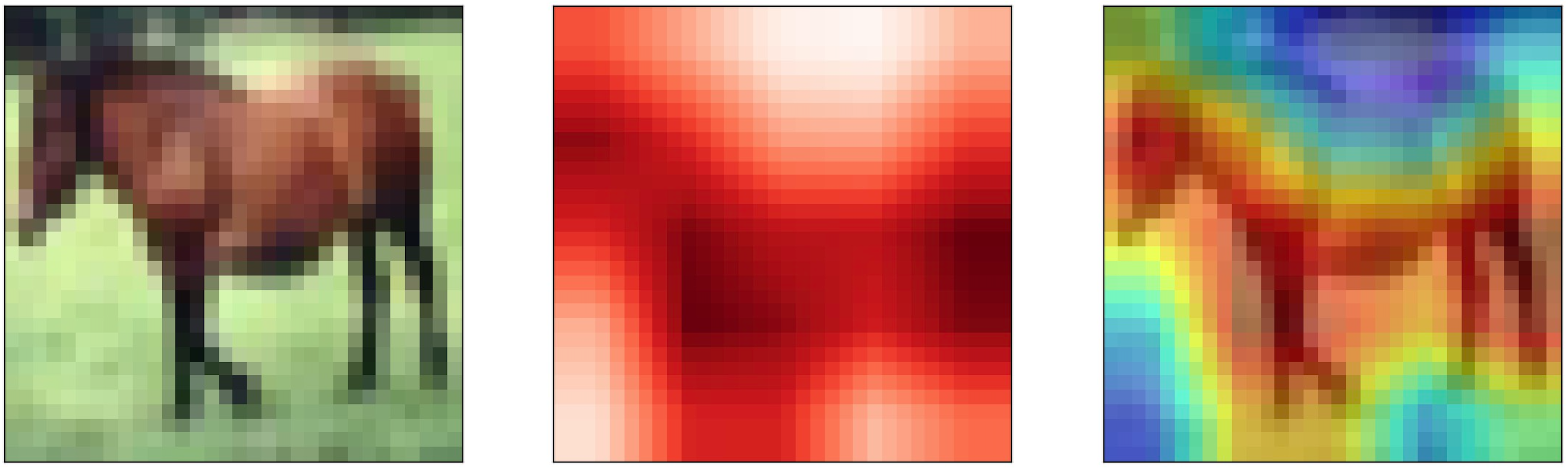}\hfil
	\includegraphics[width = 0.43\textwidth]{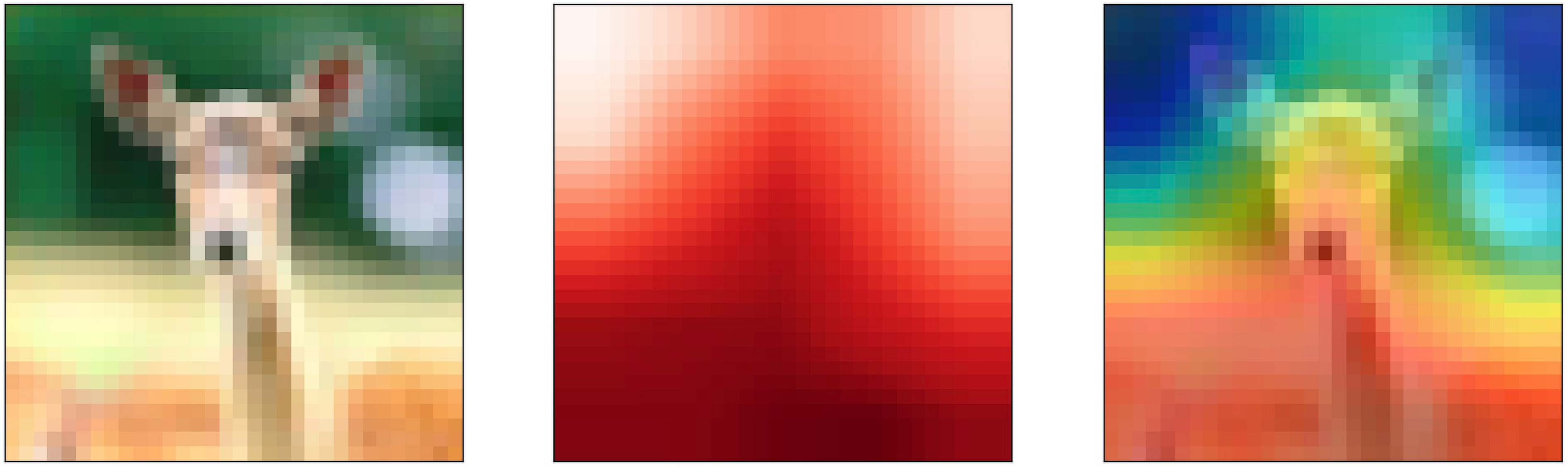}
	\caption{Heat Map for Visual Representation via $\dphero$ SGD}
	\label{heat_dog}
\end{figure}

\begin{figure}[!h]
	\centering
	\includegraphics[width = 0.43\textwidth]{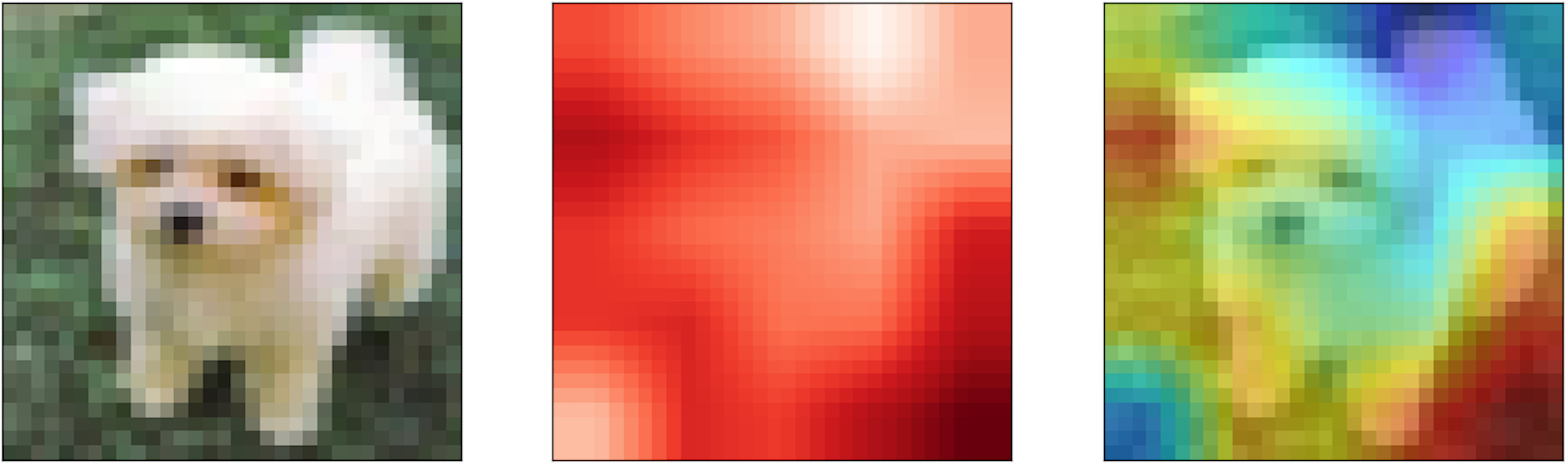}\hfil
	\includegraphics[width = 0.43\textwidth]{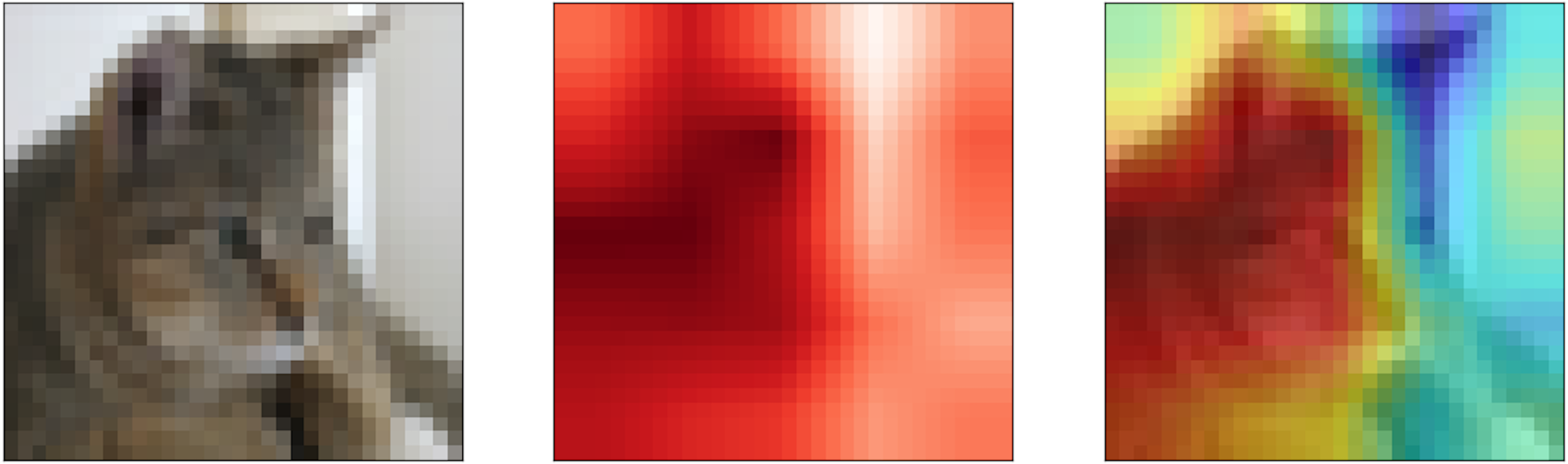}
	\includegraphics[width = 0.43\textwidth]{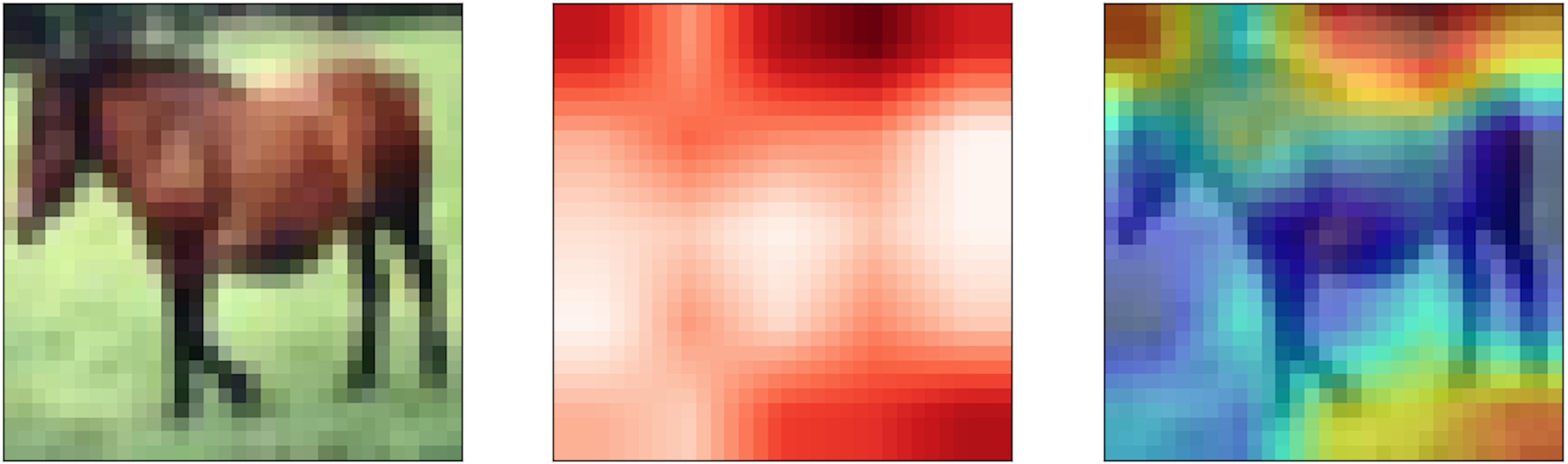}\hfil
	\includegraphics[width = 0.43\textwidth]{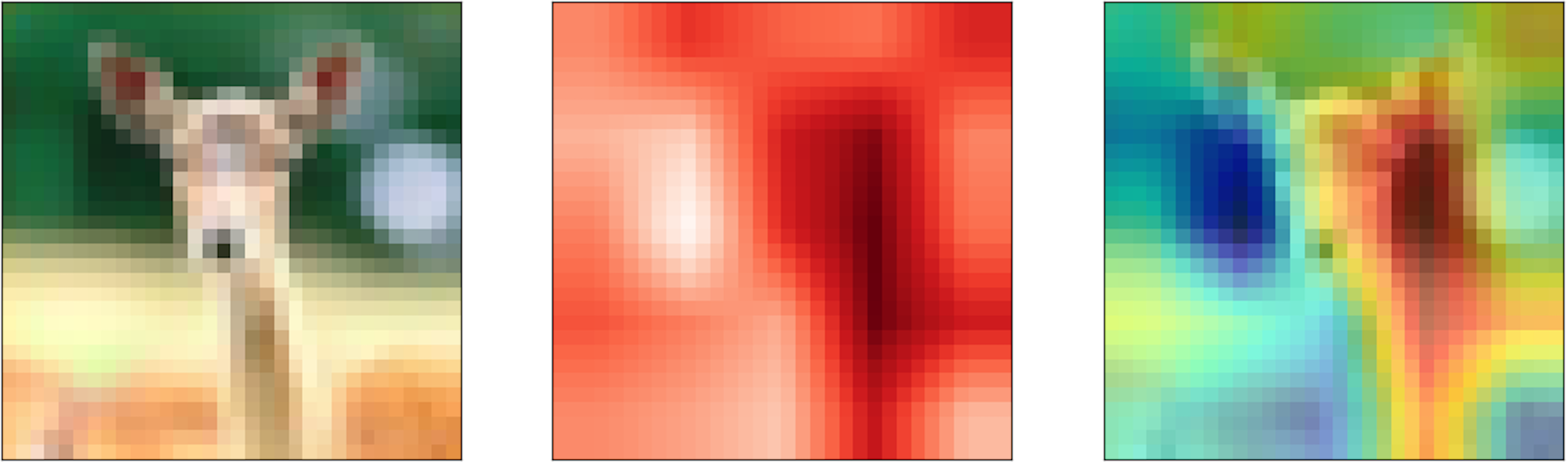}
	\caption{Heat Map for Visual Representation via DP-SGD}
	\label{heat_dog_nonoise}
\end{figure}

\subsubsection{Understanding of Improved Model Convergence}
Motivated by utility improvement, we perform repeated experiments similar to $\S$~\ref{sec::model_utility} to attain the relation between model training and  noise heterogeneity in empiricism.
We repeatedly train an identical model given various heterogeneity (adjust noise scales to diverse model parameters for early-stage tests) and witness the corresponding phenomenon in the convergence process.
Pure SGD training could attain the best accuracy and converge fastest while training with DP-SGD slows down the convergence constrained by identical remaining configurations.
Even after the model's convergence, DP-SGD training can not reach the highest accuracy as pure SGD training. 

For testing the $\dphero$ SGD, we adjust  noise allocations via PCA by injecting them into different model parameters and gradients within an identical privacy budget constraint. 
Accordingly, we could attain some convergence statuses that show lower convergence performance yet better than DP-SGD.
In practical training, utility loss can be interpreted to be convergence retard and degrading accuracy.
Improving model utility could be explained as follows: Given an identical privacy budget, a feasible solution can always exist in a region that is upper-bounded by the ground truth and lower-bounded by fixed noise perturbation.

\subsection{Further Discussion}
We explore the limitations of our work and point out the potential future works below.

\noindent
\textit{1). Speed up $\dphero$ SGD.} We observe the computation costs of PCA over a large parameter matrix are not lightweight enough.
The computational cost for $\mathsf{SVec}$ relies on the size of  the inputting matrix. 
The block-wise computation may simplify initializing a full-rank matrix as basis vectors.
Partitioning the parameter matrix into multiple blocks could speed up training in parallel; however, it may hurt the pre-existing on-the-whole knowledge stored in the current model.
Another direction is to consider a computation-light method of extracting the pre-existing knowledge learned in the current model.

\noindent
\textit{2). Architecture-specified construction.} 
To acquire a new perspective of improving model utility, the proposed construction is a feasible solution  but is not optimal.  
Although the trainable model could be regarded as a representation of knowledge extracted from diverse features and private data, different parameters are structured with the constraint of model initialization. 
At each backpropagation, we regard the model as a matrix in which each entry feeds with the values of model parameters, overlooking the effect of model structure.
In the future, instead of a generic solution, we would like to explore an architecture-specified construction of $\dphero$ SGD.


 \section{Related Works}
 \label{app::related}
 \subsection{Differential Privacy for Deep Learning} 
 Differential privacy has emerged as a solid solution to safeguard privacy in the field of deep learning.
Differential privacy (DP) for deep learning can be classified into four directions: input perturbation~\cite{input,dpsgd1}, output perturbation~\cite{iclr/PapernotSMRTE18}, objective perturbation \cite{objective1,sp/IyengarNSTTW19}, and utility optimization~\cite{corr/abs-1812-06210,ccs/AbadiCGMMT016,nips/ChenWH20,sp/Yu0PGT19}, showing valuable insights in the aspects of theory and practice.
DP could quantify to what extent individual privacy (\ie, whether a data point contributes to the training model) in a statistical dataset is preserved while releasing the established model trained over the specific dataset.
Typically, DP learning has been accomplished by applying the unbiased Gaussian noise to the gradient descent updates, a notable example of DP-SGD~\cite{ccs/AbadiCGMMT016}.
To be specific, DP-SGD adds the i.i.d. noise sampled from Gaussian distribution to model gradients to protect example-level training data involved in the training process in every iteration.

The noise-adding mechanism has been widely adopted in various  learning algorithms, e.g., image classification and natural language processing.
PATE~\cite{iclr/PapernotSMRTE18} is an approach to providing differentially private aggregation of a teacher-student model.
Due to the property of post-processing~\cite{fttcs/DworkR14}, the student's model is differentially private since it trains over the noisy inputs.
Bayesian differential privacy~\cite{icml/TriastcynF20} takes into account the data distribution for practicality~\cite{nips/JagielskiUO20}. 
By instantiating hypothetical adversaries~\cite{sp/NasrSTPC21}, various threat models are employed to show corresponding levels of privacy leakage from both the views of practitioners and theoreticians.

Privacy auditions and attacks, or cryptographic protection belong to  orthogonal research directions, focusing on the evaluative estimation of the privacy guarantee or cipher-text transmission.
Membership inference attack~\cite{ccs/0001MMBS22} enables detecting the presence or absence of an individual example, implying a lower bound on the privacy parameter $\epsilon$ via statistics~\cite{nips/EsmaeiliMPST21}.
Notably, Steinke~\etal~\cite{nips/0002NJ23} theoretically proves the feasibility of auditing privacy through membership inference on multiple examples simultaneously, elaborating an efficient one-round solution.
Combining different techniques with this work can be promising, while it is out of scope for this work.

\subsection{Privacy-Utility Trade-off}
For acquiring higher utility~\cite{pvldb/SchalerHS23},  recent works explore the adaptive mechanism of DP training from different perspectives.
They try to either reallocate/optimize the privacy budget~\cite{kdd/LeeK18,ccs/MohammadyXHZ0PD20,sp/Yu0PGT19,tit/GengV16,osdi/LuoPTCGL21} or modify the clip-norms~\cite{corr/abs-1908-07643,corr/abs-1812-02890}  of a (set of) fixed noise distribution(s) in each iteration.
Such a branch of work points out a promising direction of adaptivity via redesigning the randomness.
Privacy budget scheduling~\cite{osdi/LuoPTCGL21} improves the utility of differentially private algorithms in various scenarios.
Unlike the aforementioned advances of dynamic noise allocation, our exploration of adjusting noise heterogeneity by model parameters aims to improve the utility of the established model at every iteration rather than optimizing the noise allocation in the range of the whole training process with a constant budget.
Handcrafted features, learned from public data, can improve model utility given a fixed privacy budget \cite{iclr/TramerB21}.
Rather than introducing auxiliary data, we aim  to extract knowledge from protected model parameters without extra data assistance.

Previous analyses have enabled an understanding of utility bounds for DP-SGD mainly in an empirical manner.
Altschuler and Talwar~\cite{nips/AltschulerT22} explored the theory foundation of privacy loss --  how sensitive the output of DP-SGD is.
They solve a tighter utility bound given the privacy loss as a function of the number of iterations, concluding that  after a small burn-in period, running DP-SGD longer leaks no further
privacy. 
In this work, we exploit visual explanation~\cite{iccv/SelvarajuCDVPB17} and theoretical understanding to explore the essence of privacy-utility space.

\section{Conclusion} 
\label{sec::conclusion}
Through theoretical and empirical understanding of privacy-utility space, we extend the research line of improving training performance for DP learning by designing a plug-in optimization for training with DP-SGD.
The proposed $\dphero$ is a versatile differential privacy framework incorporating the heterogeneous DP noise.
The primary innovation of $\dphero$ is its ability to utilize the knowledge embedded in previously trained models to guide the subsequent distribution of noise heterogeneity, thereby optimizing its utility. 
Building on the foundation of $\dphero$, we introduce a heterogeneous version of DP-SGD, in which the noise introduced into the gradients varies. We have carried out extensive experiments to validate and elucidate the efficacy of $\dphero$. 
Accordingly, we provide insights on enhancing the privacy-utility space by learning from the pre-existing leaked knowledge encapsulated in the previously trained models.
Broadly, we point out a new way of thinking about model-guided noise allocation for optimizing SGD-dominated convergence under the DP guarantee.
Besides, we explore explaining DP training via visual representation, reasoning the improved utility.
Such an explainable view could benefit from understanding DP protection more vividly, for potentially being against attacks better.

\bibliography{aaai2026}


\end{document}